\documentclass[journal]{IEEEtran}

\ifCLASSINFOpdf
 
\else
 
\fi

\usepackage{cite}
\usepackage{amsmath}
\usepackage{amssymb}
\usepackage{amsthm} 
\usepackage{caption}
\usepackage{subcaption}
\usepackage{epstopdf}
\usepackage{booktabs}
\usepackage[table,xcdraw]{xcolor}
\usepackage{graphicx}
\usepackage{float}
\usepackage{color}
\usepackage{latexsym}
\usepackage{algorithm}
\usepackage{multicol,lipsum}
\usepackage[nodisplayskipstretch]{setspace}
\usepackage{optidef}
\usepackage[noend]{algpseudocode}
\usepackage{amsmath}
\usepackage{tikz}
\usepackage{mathdots}
\usepackage{yhmath}
\usepackage{cancel}
\usepackage{color}
\usepackage{siunitx}
\usepackage{array}
\usepackage{multirow}
\usepackage{amssymb}
\usepackage{gensymb}
\usepackage{tabularx}
\usepackage{booktabs}
\usetikzlibrary{fadings}
\usetikzlibrary{patterns}
\usetikzlibrary{shadows.blur}
\usetikzlibrary{shapes}
\usepackage{lipsum}
\usepackage{mathtools}
\usepackage{cuted}
\usepackage{algorithm}
\usepackage{algpseudocode}
\usepackage{pifont}

\newtheorem{theorem}{Theorem}
\newtheorem{remark}{Remark}
\newtheorem{property}{Property}
\newtheorem{lemma}{Lemma}

\begin{document}
	
%
\title{Percentile Optimization in Wireless Networks—\\Part I: Power Control for Max-Min-Rate to Sum-Rate Maximization (and Everything in Between)}
%
%
%

\author{Ahmad~Ali~Khan,~\IEEEmembership{Student Member,~IEEE,}
        Raviraj~S.~Adve,~\IEEEmembership{Fellow,~IEEE.}\vspace{-2.7em}
\thanks{Raviraj Adve is with the Edward S. Rogers Sr. Department
of Electrical and Computer Engineering, University of Toronto, Toronto,
Ontario, M5S 3G4, e-mail: rsadve@ece.utoronto.ca. Ahmad Ali Khan is with Ericsson R\&D, Ottawa, Ontario, K2K 2V6, email: ahmad.a.khan@ericsson.com. The authors would like to acknowledge the support of Ericsson Canada.}}

\markboth{IEEE TRANSACTIONS ON SIGNAL PROCESSING (DRAFT)}%
{Khan and Adve \MakeLowercase{\textit{et al.}}: Percentile Optimization in Wireless Networks-- Part 1.}

%



\maketitle

\begin{abstract}
  Improving throughput for cell-edge users through coordinated resource allocation has been a long-standing driver of research in wireless cellular networks. While a variety of wireless resource management problems focus on sum utility, max-min utility and proportional fair utility, these formulations do not explicitly cater to cell-edge users and can, in fact, be disadvantageous to them. In this two-part paper series, we introduce a new class of optimization problems called percentile programs, which allow us to explicitly formulate problems that target lower-percentile throughput optimization for cell-edge users. Part I focuses on the class of least-percentile throughput maximization through power control. This class subsumes the well-known max-min and max-sum-rate optimization problems as special cases. Apart from these two extremes, we show that least-percentile rate programs are non-convex, non-smooth and strongly NP-hard in general for multiuser interference networks, making optimization extremely challenging. We propose cyclic maximization algorithms that transform the original problems into equivalent block-concave forms, thereby enabling guaranteed convergence to stationary points. Comparisons with state-of-the-art optimization algorithms such as successive convex approximation and sequential quadratic programming reveal that our proposed algorithms achieve superior performance while computing solutions orders of magnitude faster. 
 
\end{abstract}

\begin{IEEEkeywords}
Percentile optimization, cell-edge, power control, cyclic maximization, least-percentile rate, greatest-percentile.
\end{IEEEkeywords}

%
\IEEEpeerreviewmaketitle

\section{Introduction}
%
%
%
%
\subsection{Background and Overview}
\IEEEPARstart{C}{oordinated} resource allocation through optimization is an area of considerable research interest due to its potential for improving service in wireless cellular networks ~\cite{bjornson_optimal_2013,gesbert_adaptation_2007, soret_interference_2018}. Due to the randomness and inherent physical characteristics of the wireless propagation medium, however, there is significant variation in the throughput, energy efficiency and power consumption achieved by different users across a cellular network \cite{ma_interference-alignment_2018}. {\color{black} In particular, users located close to the boundaries of cells (commonly referred to as \textit{cell-edge} users) consistently experience poor received signal strength from their serving base station (BS) as well as strong interference from neighboring cells \cite{zhang_weighted_2011,yang_cell-edge-aware_2017,ma_interference-alignment_2018}, as illustrated in Figure \ref{edge_users}. This leads to poor achieved data rates for cell-edge users manifested in the lower percentiles (e.g., $5^\mathrm{th}$- and $10^\mathrm{th}$-percentile rates). Improving these lower-percentile rates has been identified as a crucial requirement for current and next-generation cellular networks to enhance access as well as enabling new applications \cite{3gpp,5gppp,itu}. For example, the 3GPP industrial standards for future wireless cellular networks call for a 3$\times$ improvement in $5^\mathrm{th}$-percentile rates as compared to the current 5G levels \cite{3gpp}; similar targets have been proposed in joint industry-academic studies carried out by the 5G Public Private Partnership (5GPPP) \cite{5gppp}. Likewise, a recent work by a major cellular network infrastructure provider suggests a target of 5 Gbps for the $5^\mathrm{th}$-percentile rate in 6G networks \cite{ziegler_6g_2020}. As such, numerous research papers have been published by equipment vendors \cite{sihlbom_reconfigurable_2022,ziegler_6g_2020} exploring how increased \textit{physical resources}, such as number of transmission antennas and orthogonal frequency bands, can help achieve these targets for lower-percentile rates.}

{\color{black}Despite the clear necessity of improving lower-percentile service, prior research works have, at best, proposed indirect techniques for tackling this goal. Physical-layer techniques like fractional frequency reuse \cite{chang_optimal_2016}, heterogeneous architectures \cite{mankar_load-aware_2018} and densification \cite{lopez-perez_towards_2015} have been explored in prior works with the aim of improving lower-percentile cell-edge rates. While effective, these approaches are nonetheless heuristic in nature as they do not directly address the mathematical problem of optimizing throughput for a desired percentile. Additionally, these techniques exhibit various undesirable characteristics; fractional frequency reuse schemes, for example, are bandwidth-inefficient since the frequency allocations are static within each cell, leading to under-utilization of the available spectrum. Additionally, there is no optimal method to assign the frequency tones among different users. It is critical to recognize that none of these aforementioned methods are based on signal processing techniques.}
\begin{figure}[t!]
	\begin{center} 
		\includegraphics[trim={3cm 3cm 3cm 3cm},clip, width=0.40\textwidth]{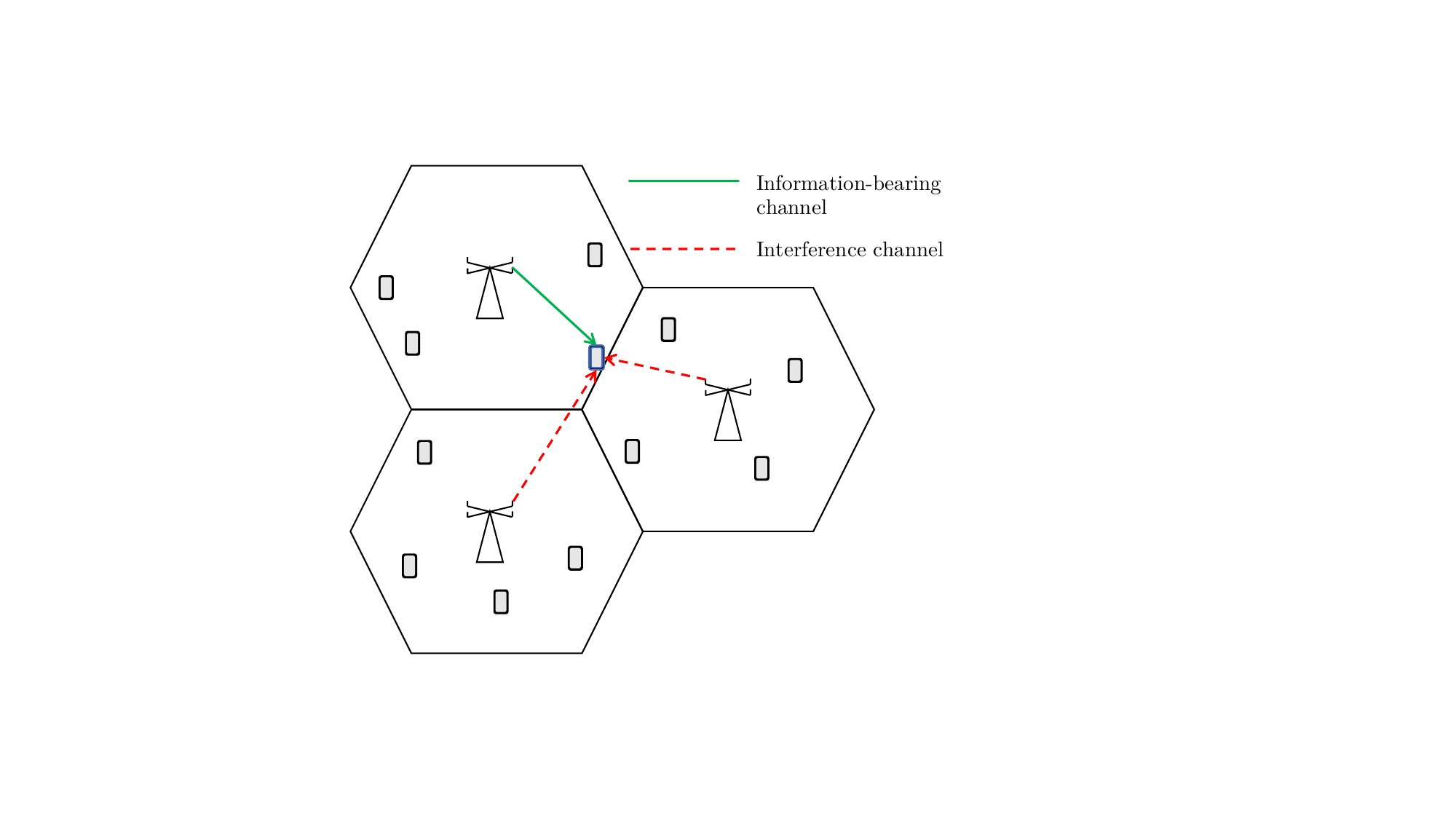}
		\caption{{\color{black}A typical cell-edge user in a wireless cellular network.}}
		\centering
		\label{edge_users}
	\end{center}
\end{figure}

{\color{black}In contrast to industrial targets, the $10^\mathrm{th}$-percentile throughput is widely used as a measure of cell-edge rate in most academic works including, but not limited to, \cite{shen_fractional_2018-1,khan_optimizing_2020,naderializadeh_resource_2021}. However, prior research works have come up with \textit{indirect} ways to optimize this metric via signal processing techniques like beamforming and power control. This is usually accomplished by solving a network utility maximization problem, with the most common choices of network utility function being sum-rate \cite{zhi-quan_luo_dynamic_2008,evangelista_fairness_2019}, minimum rate \cite{zhi-quan_luo_dynamic_2008,razaviyayn_linear_2013,naghsh_maxmin_2019} and proportional fair (PF) (usually approximated as weighted sum-rate (WSR)) \cite{shi_iteratively_2011,shen_fractional_2018,khan_optimizing_2018,zhang_weighted_2011}; however, none of these utilities explicitly targets cell-edge users.}

Sum-rate maximization is well-known to be unfair to edge users: as noted in \cite{rahman_interference_2009}, transmitting to edge users creates interference for cell-center users, driving the network throughput down. Hence, optimal solutions to the sum-rate maximization problem assign little to no power to these users \cite{shen_fractional_2018-1}. In contrast, the minimum rate caters to the \textit{weakest edge user} in the network resulting in a large rate penalty for most other users. This phenomenon makes it an unsuitable utility for large networks: as noted by \cite{ghazanfari_enhanced_2020}, the minimum achievable rate in a multiuser network asymptotically approaches zero as the number of users grows. 

Sum-log-average-rate optimization (and the corresponding WSR maximization problem) on the other hand, is proportionally fair (PF) and aims to serve all users. Nevertheless, we emphasize that WSR maximization \textit{does not optimize the edge user rates}; essentially, this is an ad-hoc choice of metric. Hence, any comparison of edge rates achieved by different WSR maximization algorithms, such as the weighted minimum mean-squared error (WMMSE) scheme proposed in \cite{shi_iteratively_2011}, is not mathematically rigorous. Schemes that achieve higher sum-log-utility do not necessarily achieve higher edge rates and vice versa; as a specific example, the approach presented in \cite{khan_optimizing_2018} improves sum-log-utility over the WMMSE algorithm, but results in lower $10^\mathrm{th}$-percentile rates.

Alternative utilities like $\alpha$-fair power allocation \cite{guo_alpha_2016} and Jain's fairness index \cite{sediq_optimal_2013} have also been proposed to prioritize serving cell-edge users, but as observed in \cite{lan_axiomatic_2010}, such fairness measures are not \textit{physically interpretable} as we cannot distinguish whether one measure is `better' than another. Enforcing constraints on individual throughputs is also a similar heuristic to enable fairer resource allocation \cite{eisen_learning_2019,li_multicell_2021}.

Finally, we note that cell-edge throughput optimization has been attempted indirectly under the framework of stochastic analysis \cite{hosseini_optimizing_2018,yang_cell-edge-aware_2017}. A resource allocation scheme is chosen and closed-form expressions for distributions of the desired metric are obtained using stochastic geometry. The resource allocation parameters that yield the highest cell-edge throughput for a desired percentile are then empirically chosen \cite{hosseini_optimizing_2018,yang_cell-edge-aware_2017}. It should be emphasized, however, that such schemes enable \textit{analysis} but not \textit{direct optimization} of cell-edge throughput as the power control and precoding strategies are essentially pre-determined.

In summary, despite the crucial importance of throughput optimization at desired percentiles for future generations of wireless networks, none of the prior resource allocation works (including the aforementioned FP \cite{shen_fractional_2018-1} and WMMSE methods \cite{shi_iteratively_2011}) have tackled this class of problems. Therefore, we believe that there is a clear need to study these problems as a result of their practical significance as well as the lack of a rigorous theoretical framework to solve them in the available literature. Motivated by the practical significance, theoretical challenges, and the lack of a coherent framework, this two-part paper series addresses this critical gap by formulating percentile programs in wireless networks and showcases iterative optimization methods to solve them effectively.

\subsection{Contributions of Part I}
In this paper, we introduce a new class of optimization problem with specific reference to resource allocation problems in communications, which we henceforth term as \textit{percentile programs}. Part I of this two-paper series focuses on sum-percentile rate optimization problems. Specifically, the contributions of this part can be summarized as follows:

\begin{itemize}
	\item{\textbf{Problem formulation:} We develop an explicit problem formulation for the maximization of sum-least-percentile rate via power control for a multicell, multiuser single-input, single-output (MU-SISO) network for any choice of percentile. The proposed formulation enables us to directly optimize the throughput for any choice of percentile; to the best of our knowledge, our work is the first of its kind to directly tackle this class of problems. 
		
	{\color{black}\item{\textbf{Optimizing cell-edge rates:}} This formulation directly targets cell-edge optimization by targeting lower percentiles, e.g. $5^\mathrm{th}$-percentile rate, in accordance with the industrial targets for next-generation networks \cite{3gpp,5gppp,itu,ziegler_6g_2020}. Additionally, we show that the choice of the percentile level allows us to control the tradeoff between favoring cell-center (i.e., sum-rate) users or cell-edge users.}}
	
	\item{\textbf{Complexity analysis:} Through a polynomial-time reduction from the maximum independent set problem, we demonstrate that the power-control sum-percentile rate optimization problem is strongly NP-hard, in general, for all percentiles strictly greater than the minimum percentile. In this regard, this work is the first of its kind to generalize the previous results that were restricted to the max-min-rate and sum-rate problems \cite{zhi-quan_luo_dynamic_2008}.}
	
	\item{\textbf{Cyclic optimization to stationarity:} We transform the original sum-least-percentile utility into equivalent block-concave forms, yielding convenient cyclic maximization algorithms. Furthermore, we demonstrate that the proposed algorithms belong to the broader class of minorization-maximization algorithms, and present proofs of non-decreasing convergence to stationary points of the original problem.}
		
\end{itemize}

\subsection{Notation}
Prior to proceeding further, we define some notation used in this paper. $\mathbb{R}$, ${\mathbb{R}}_{{+}}$ and ${\mathbb{R}}_{{++}}$ represent the set of real numbers, non-negative real numbers, and positive numbers respectively. {\color{black}{Similarly, we use $\mathbb{Z}$, ${\mathbb{Z}}_{{+}}$ and ${\mathbb{Z}}_{{++}}$ to denote the set of integers, non-negative integers, and positive integers respectively.}} We denote scalars using lowercase letters (e.g., $x$), vectors using lowercase boldface (e.g., $\mathbf{x}$), matrices using uppercase boldface (e.g., $\mathbf{X}$) and sets using script typeface (e.g., $\mathcal{X}$). The optimal value of a variable is denoted by a $*$ sign in superscript (e.g., $x^{*}$, $\mathbf{x}^{*}$, $\mathbf{X}^{*}$). We denote the all-ones and all-zeros vectors by $\mathbf{1}$ and $\mathbf{0}$, respectively, with dimension determined by the context. The operator $\left|\cdot\right|$ denotes the absolute value when applied to a scalar and cardinality when applied to a set. 

For a vector $\mathbf{x}\in\mathbb{R}^{n}$, we denote by $\mathbf{x}^{\downarrow}$ the vector with identical components, but sorted in descending order; likewise, we denote by $\mathbf{x}^{\uparrow}$ the vector with identical components but sorted in ascending order. Hence, $x_{i}^{\downarrow}$ ($x_{i}^{\uparrow}$) then denotes the $i^\mathrm{th}$ entry in $\mathbf{x}^{\downarrow}$ ($\mathbf{x}^{\uparrow}$), i.e., the $i^\mathrm{th}$ largest (smallest) entry in $\mathbf{x}$. The number of ways of picking $k$ unordered outcomes from $n$ possibilities (i.e., the binomial coefficient) is denoted by $_n{C}_k$.

\subsection{Organization}
Part I is organized as follows:
\begin{itemize}
	\item In Section \ref{properties}, we define sum-percentile functions and discuss key properties needed to construct the subsequent optimization problems.
	\item We formulate and solve the sum-least percentile throughput optimization problem for the special setting of parallel Gaussian channels in Section \ref{gaussian}. 
	\item Section \ref{interference_channel} deals with the more realistic setting of an interference-limited multicell, multiuser network. Finally, we draw conclusions in Section \ref{conclusions}.
\end{itemize}		
To maintain consistency, Sections \ref{gaussian} and \ref{interference_channel} follow an identical order: we first introduce the problem formulation, then describe the proposed approach(es) and finally present the numerical results at the end of section. Thus, the numerical results for the simplified parallel Gaussian channel setting are given in Section \ref{pg_results} while those for the more realistic interference-limited cellular network are given in Section \ref{intf_results}.

Part II deals with the more general setting of beamforming for a multiuser multiple-input multiple-output (MIMO) network, as well as the construction of novel utility functions by composition, and ergodic percentile throughput optimization.
\section{Sum-Percentile Functions and Their Properties}\label{properties}

\subsection{Definitions}
We begin by defining two classes of sum-percentile utility functions. Given an integer $K\in\mathbb{Z_{++}}$ and a percentile $q\in\left(0,100\right]$, we define the corresponding \textit{percentile number} $K_q$ as
\[
K_{q}=\mathrm{min}\left\{ k\in\mathbb{Z_{++}}\left|\frac{100k}{K}\geq q\right.\right\} 
\]

Given the vector $\mathbf{x}=[x_{1},\ldots,x_{K}]^{T}$, we define the sum-greatest-$q^\mathrm{th}$-percentile (SGqP) utility function $\ensuremath{F_{K_{q}}\left(\mathbf{x}\right)}:\mathbb{R}^{K}\mapsto\mathbb{R}$ as the sum of the $K_q$ largest components of $\mathbf{x}$, i.e.,
\begin{equation}\label{SGqP_definition}
F_{K_{q}}\left(\mathbf{x}\right)=\sum_{i=1}^{K_{q}}x_{i}^{\downarrow}
\end{equation}

Similarly, we define the sum-least-$q^\mathrm{th}$-percentile (SLqP) utility function $\ensuremath{f_{K_{q}}\left(\mathbf{x}\right)}:\mathbb{R}^{K}\mapsto\mathbb{R}$  as the sum of the smallest $K_q$ components of $\mathbf{x}$, i.e.,
\begin{equation}\label{SLqP_definition}
f_{K_{q}}\left(\mathbf{x}\right)=\sum_{i=1}^{K_{q}}x_{i}^{\uparrow}
\end{equation}

We shall often find it convenient to explicitly indicate the components of a vector when dealing with SLqP and SGqP functions; thus, throughout this paper, $f_{K_{q}}\left(\mathbf{x}\right)$ and $f_{K_{q}}\left(x_{1},\ldots,x_{n}\right)$ are equivalent.

\subsection{Properties}
We now state some key properties of the SLqP and SGqP utility functions. The latter class of functions is dealt with only in Part II and may be safely skipped for readers wishing to focus on the power control problems examined in Part I.

\begin{property}\label{convexity_concavity}
	The SGqP (resp. SLqP) utility function is convex (resp. concave) for all values of $q\in\left(0,100\right]$.
\end{property}
\begin{proof}
	Observe that we can express the SGqP utility function in an alternative form as follows. We define the following set:
	\begin{equation}\label{binary_vectors_set}
	\mathcal{A}_{K_{q}}\coloneqq\left\{ \left.\mathbf{a}\in\left\{ 0,1\right\} ^{K}\right|\mathbf{1}^{T}\mathbf{a}=K_{q}\right\}
	\end{equation}
	
	In other words, $\mathcal{A}_{K_{q}}$ is the set of binary vectors of length $K$ with Hamming weight of $K_q$; this set has cardinality $_{K}C_{K_q}$. Then the SGqP utility function can be equivalently expressed as
	\begin{equation}\label{SGqP_pointwise_maximum}
	F_{K_{q}}\left(\mathbf{x}\right)=\underset{\mathbf{a}\in\mathcal{A}_{K_{q}}}{\mathrm{max}}\hspace{0.50em}\mathbf{a}^{T}\mathbf{x}
	\end{equation}
	which is the pointwise maximum over a set of functions that are linear in $\mathbf{x}$, and is thus convex \cite{boyd_convex_2004}. Next, observe that we can similarly write the SLqP utility function equivalently as: 
	\begin{equation}
	f_{K_{q}}\left(\mathbf{x}\right)=\mathbf{1}^{T}\mathbf{x}-F_{K-K_{q}}\left(\mathbf{x}\right)
	\end{equation}
	which is the difference of a linear and a convex function, and thus concave overall in $\mathbf{x}$ \cite{boyd_convex_2004}.
\end{proof}

\begin{property}\label{SLqP_SGqP_nondecreasing}
The SGqP and SLqP utility functions are non-decreasing in each component.
\end{property}
\begin{proof}
For a fixed value of $q$, the SGqP (resp. SLqP) utility function is the pointwise maximum (resp. pointwise minimum) of non-decreasing linear functions and is therefore also non-decreasing.
\end{proof}

\begin{property}\label{special_cases_SGqP}
The SGqP utility function contains as special cases the sum and maximum utility functions.
\end{property}
\begin{proof}
	Observe that if $q=100$, we have $K=K_q$, in which case the SGqP utility function simply becomes the sum of all components of the vector:
	\[
F_{K}\left(\mathbf{x}\right)=\sum_{i=1}^{K}x_{i}^{\downarrow}=\mathbf{1}^{T}\mathbf{x}^{\downarrow}=\mathbf{1}^{T}\mathbf{x}
	\]
	Similarly, if we set $q=100/K$, we have $K_q=1$ and 
	\[
F_{1}\left(\mathbf{x}\right)=x_{1}^{\downarrow}=\underset{i=1,\ldots,K}{\mathrm{max}}x_{i}
	\]\end{proof}

\begin{property}\label{special_cases_SLqP}
	The SLqP utility function contains as special cases the sum and minimum utility functions.
\end{property}
\begin{proof}
	This property follows from similar reasoning to Property \ref{special_cases_SGqP}; we have
	\[
f_{K}\left(\mathbf{x}\right)=\sum_{i=1}^{K}x_{i}^{\uparrow}=\mathbf{1}^{T}\mathbf{x}^{\uparrow}=\mathbf{1}^{T}\mathbf{x}
	\]
	\[
f_{1}\left(\mathbf{x}\right)=x_{1}^{\uparrow}=\underset{i=1,\ldots,K}{\mathrm{min}}x_{i}
	\]
\end{proof}
\vspace{-1.00em}
\begin{property} \label{nonsmoothness_SGqP_SLqP}
The SGqP and SLqP utility functions are non-smooth for $K_q\neq{K}$.
\end{property}
\begin{proof}
This property is readily apparent from the form of the SGqP utility function in (\ref{SGqP_pointwise_maximum}); since it is the pointwise maximum of $_{K}C_{K_q}$ distinct linear functions, it follows that the function, although continuous, is not necessarily smooth \cite{todd_max-k-sums_2018}.\end{proof} 

We note that the $q=100$ case is an exception to Property \ref{nonsmoothness_SGqP_SLqP}: for $K_q=K$, both the SGqP and SLqP functions are equivalent to the sum utility function, which is smooth. 

\begin{property} The SGqP and SLqP functions are connected through the following symmetry property:
	\[F_{K_{q}}\left(\mathbf{x}\right)=-f_{K_{q}}\left(\mathbf{-x}\right)\]
\end{property}
\begin{proof}
	Following the definition of $\mathcal{A}_{K_{q}}$ in (\ref{binary_vectors_set}), we observe that 
	\[
	F_{K_{q}}\left(-\mathbf{y}\right)=\underset{\mathbf{a}\in\mathcal{A}_{K_{q}}}{\mathrm{max}}\hspace{0.50em}-\mathbf{a}^{T}\mathbf{y}=-\underset{\mathbf{a}\in\mathcal{A}_{K_{q}}}{\mathrm{min}}\hspace{0.50em}\mathbf{a}^{T}\mathbf{y}=-f_{K_{q}}\left(\mathbf{y}\right)
	\]
	The property follows from the change of variables $\mathbf{y}=-\mathbf{x}$.
\end{proof}
\begin{property}\label{ordering}
	For all $\mathbf{x}\in\mathbb{R}_{+}^{K}$, the SGqP and SLqP functions are ordered by the percentile number: for $K_{q_1}<K_{{q_2}}$, we have $f_{K_{q_1}}\left(\mathbf{x}\right)\leq f_{K_{{q_2}}}\left(\mathbf{x}\right)$ and $F_{K_{q_1}}\left(\mathbf{x}\right)\leq F_{K_{{q_2}}}\left(\mathbf{x}\right)$.
\end{property}
\begin{proof}
We illustrate this property for the SLqP function; the proof for the SGqP function is similar. Suppose that $K_{q_1}+i=K_{{q_2}}$ where $i\in\mathbb{Z_{++}}$. Then, from the definition of the SLqP function in (\ref{SLqP_definition}) we can write:
\[
f_{K_{q_{2}}}\left(\mathbf{x}\right)=f_{K_{q_{1}}}\left(\mathbf{x}\right)+x_{K_{q_{1}}+1}^{\uparrow}+\ldots+x_{K_{q_{1}}+i}^{\uparrow}
\]
The property then follows from the fact that $x_{k}^{\uparrow}\geq{0}$ for $k=1,\ldots,K$.
\end{proof}

To summarize, the SGqP and SLqP functions are non-smooth utilities. Due to their respective convexity and concavity, we can formulate minimization problems for the former and maximization problems for the latter utility. 

\section{Optimization for Parallel Gaussian Channels}\label{gaussian}
\subsection{Problem Formulation}
{\color{black}We are now ready to formally introduce the class of sum-least-$q^\mathrm{th}$-percentile rate maximization problems. We begin by considering the problem of SLqP rate maximization under a interference-free parallel channel model with unity channel gains. We study this setting for two key reasons. First, percentile rate problems have never been studied in this particular setting in the prior literature (other than the special max-min-rate and sum-rate problems \cite{zhi-quan_luo_dynamic_2008}). Second, the problem is convex in this simplified setting, enabling derivation of many useful properties of the optimal solution structure that are not possible in the more general non-convex and NP-hard interference-limited scenario for real-world cellular networks.}

We assume a total of $K$ transmitter-receiver pairs; for convenience, we refer to each pair as a user throughout this paper. Further, we denote the transmit power of user $k$ as $p_k$ and receiver additive white Gaussian noise (AWGN) power by $z_{k}$; the former is collected in the vector $\mathbf{p}$ while the latter is collected in the vector $\mathbf{z}$ for notational clarity. For simplicity of notation and without loss of generality, we subsequently assume that the users are ordered according to descending noise powers, i.e., $z_1\geq{z_2}\geq{\ldots}\geq{z_K}$. {\color{black}{In keeping with prior resource allocation works, we assume that these noise powers are known \cite{shen_fractional_2018-1,khan_optimizing_2020}}}. Assuming parallel channels with a gain of unity for all users, it follows that the rate achievable by user $k$ is given by
\begin{equation}\label{interference_free_rates}
r_{k}\left(\mathbf{p}\right)=\log\left(1+\frac{p_{k}}{z_{k}}\right)
\end{equation}

Our aim can be concisely expressed as follows: under a sum-power constraint of $P_T$ in the given setting, how should the powers be allocated to maximize the sum of the smallest $K_q$ rates in the network? The SLqP rate optimization problem that encapsulates this goal can be expressed as follows:
 \begin{subequations}\label{SPR_problem_interference_free}
 	\begin{align}
	\underset{\mathbf{p^{\mathit{}}}}{\mathrm{maximize}}\quad f_{K_{q}}\left(r_{1}\left(\mathbf{p}\right),\ldots,r_{K}\left(\mathbf{p}\right)\right)\hspace{2em}\label{SPR_obj_interference_free}\hspace{4.00em}\\
 	\mathrm{subject\,to}\quad p_{k}\geq0,\,k=1,\ldots,K\label{SPR_constraint_nonnegative}\hspace{7.30em}\\
 	\sum_{k=1}^{K}p_{k}\leq {P}_\mathrm{max}\label{SPR_constraint_total_power}\hspace{10.3em}
 	\end{align}
 \end{subequations}
Note that the order of the users here does not matter to the SlqP utility function.
\begin{remark}
	The percentile program in (\ref{SPR_problem_interference_free}) is equivalent to maximization of the sum-rate when $K_q=K$ and maximization of the minimum rate when $K_q=1$.
\end{remark}
\begin{proof}
This follows as a direct consequence of Property \ref{special_cases_SLqP}.
\end{proof}

\begin{remark}
Problem (\ref{SPR_problem_interference_free}) is a convex program.
\end{remark}
\begin{proof}
The convexity can be established using the composition rule, described in detail in \cite{boyd_convex_2004}. Consider a function $h:\mathbb{R}^{K}\mapsto\mathbb{R}$, and a sequence of functions $g_{k}:\mathbb{R}^{N}\mapsto\mathbb{R};\hspace{0.30em}k=1,\ldots,K$. Next, define the following composition:
\begin{equation}\label{composition_rule}
	h\left(g_{1}\left(\mathbf{x}\right),g_{2}\left(\mathbf{x}\right),\ldots,g_{K}\left(\mathbf{x}\right)\right):\mathbb{R}^{N}\mapsto\mathbb{R}
\end{equation}

This composition is concave in $\mathbf{x}$ if the following conditions are satisfied \cite{boyd_convex_2004}:
\begin{itemize}
		\item $h$ is concave and non-decreasing in each argument.
		\item Each of the functions $g_{k}:\mathbb{R}^{N}\mapsto\mathbb{R};\hspace{0.30em}k=1,\ldots,K$ is concave in $\mathbf{x}$.
\end{itemize}
	
We apply this rule to Problem (\ref{SPR_problem_interference_free}) as follows. First, we observe that the term inside the logarithm, $1+\frac{p_k}{z_k}$, is obviously linear (and hence concave) in $\mathbf{p}$. The logarithm function is concave and non-decreasing, hence it follows that $\mathrm{log}\left(1+\frac{p_k}{z_k}\right)$ is also concave in the powers. Thus, the functions $r_k\left(\mathbf{p}\right);i=1,\ldots,K$ are all concave in $\mathbf{p}$.
	
As established in Property \ref{SLqP_SGqP_nondecreasing}, the SLqP utility function is non-decreasing in each component. Furthermore, as shown in Property 1, it is also concave, since it is the pointwise minimum of a set of concave functions (which is also known to be concave according to \cite{boyd_convex_2004}).
	
Identifying $f_{K_q}$ as $h$, and $r_k$ as $g_k$ from \ref{composition_rule}, the concavity of the objective function in \ref{SPR_obj_interference_free} then follows straightforwardly. Since we are maximizing a concave function (subject to convex constraints), it follows that Problem (\ref{SPR_problem_interference_free}) is indeed a convex program.
\end{proof}

\begin{remark}\label{ordering_by_percentile}
	Denote by $f_{K_{q},\mathrm{opt}}$ the optimal objective value of Problem (\ref{SPR_problem_interference_free}). Then the optimal values are \textit{ordered} according to the percentile number: if $K_{q_1}<K_{{q_2}}$ then we must have $f_{K_{q_1},\mathrm{opt}}\leq f_{K_{{q_2}},\mathrm{opt}}$.
\end{remark}
\begin{proof}
	Let $\mathcal{R}\left(\mathbf{z},P_{T}\right)$ denote the set of $K$-tuples of achievable rates for all users given $\mathbf{z}$ and $P_T$, i.e.,
	\[
	\mathcal{R}\left(\mathbf{z},P_{T}\right)=\left\{ \left(r_{1}\left(\mathbf{p}\right),\ldots,r_{K}\left(\mathbf{p}\right)\right)\left|\begin{array}{c}
	\sum_{k=1}^{K}p_{k}\leq P_{T},\\
	p_{k}\geq0,k=1,\ldots,K
	\end{array}\right.\right\} 
	\]
	Then from Property \ref{ordering}, it follows that we must have 
	\begin{multline}
	\underset{\left(r_{1}\left(\mathbf{p}\right),\ldots,r_{K}\left(\mathbf{p}\right)\right)\in\mathcal{R}\left(\mathbf{z},P_{T}\right)}{\mathrm{max}}\quad f_{K_{q_1}}\left(r_{1}\left(\mathbf{p}\right),\ldots,r_{K}\left(\mathbf{p}\right)\right)\\
	\leq\underset{\left(r_{1}\left(\mathbf{p}\right),\ldots,r_{K}\left(\mathbf{p}\right)\right)\in\mathcal{R}\left(\mathbf{z},P_{T}\right)}{\mathrm{max}}\quad f_{K_{{q_2}}}\left(r_{1}\left(\mathbf{p}\right),\ldots,r_{K}\left(\mathbf{p}\right)\right)
	\end{multline}
	which yields the desired result. 
\end{proof}
	As a special case, we note that this inequality implies that the optimal sum-rate should be greater than or equal to the optimal min-rate; this subsumes the prior results presented in \cite{zhi-quan_luo_dynamic_2008}.

\begin{remark}\label{largest_K_minus_Kq_rates_equal}
Denote by $\ensuremath{r_{k,\mathrm{opt}}^{}}$ the rate achieved by the $k^\mathrm{th}$ user in the optimal solution for the problem in (\ref{SPR_problem_interference_free}). Then we must have $\ensuremath{r_{K_{q},opt}^{}=r_{K_{q}+1,opt}^{}=,\ldots,=r_{K,opt}^{}}$. 
\end{remark}
\begin{proof}
	This remark follows by contradiction. Suppose we are given an optimal solution with $r_{k,\mathrm{opt}}>r_{K_q,\mathrm{opt}}$ for some $\ensuremath{k\in\left\{K_{q}+1,K_{q}+2,\ldots,K\right\}}$. Then we can improve the objective by lowering $p_{k}$ and correspondingly increasing $p_{K_q}$ up until user $K_q$ achieves an identical rate to user $k$, thereby contradicting the optimality of the solution. \textit{Thus, in the optimal solution, the $K-K_q+1$ highest achieved rates must be identical. }
\end{proof}

A well-established illustration of the preceding remark is that the max-min-rate problem (i.e. $q=100/K$) leads to equal rates for all users.

\subsection{Solution Techniques and Implementation}\label{implementation}
Having established the convexity of Problem (\ref{SPR_problem_interference_free}), we note that there are a number of ways to numerically solve it. One method is to utilize the subgradient technique, described in detail in \cite{nesterov_subgradient_2014}; with decaying step sizes, the method is guaranteed to converge to the optimal solution. Alternatively, an easier way is to utilize the CVX package, which includes the SLqP and SGqP utilities in its standard library as the \verb#sum_smallest(x,k)# and \verb#sum_largest(x,k)# functions respectively. The package can then be utilized to numerically solve Problem (\ref{SPR_problem_interference_free}) using the standard interior-point methods \cite{CVX}.

\subsection{Numerical Results}\label{pg_results}
To better understand how the optimal solution changes for different values of the percentile parameter $q$, we construct a sample network with $K=6$ users, total power $P_\mathrm{max}=10$, and the following AWGN noise powers:
\[
\mathbf{z}=\left[0.1,\,0.05,\,250.1,\,200.4,\,5.4,\,3.7\right]^{T}
\]
\begin{figure}[t!]
	\begin{center} 
		\includegraphics[trim={0cm 0cm 0cm 0cm},clip, width=0.40\textwidth]{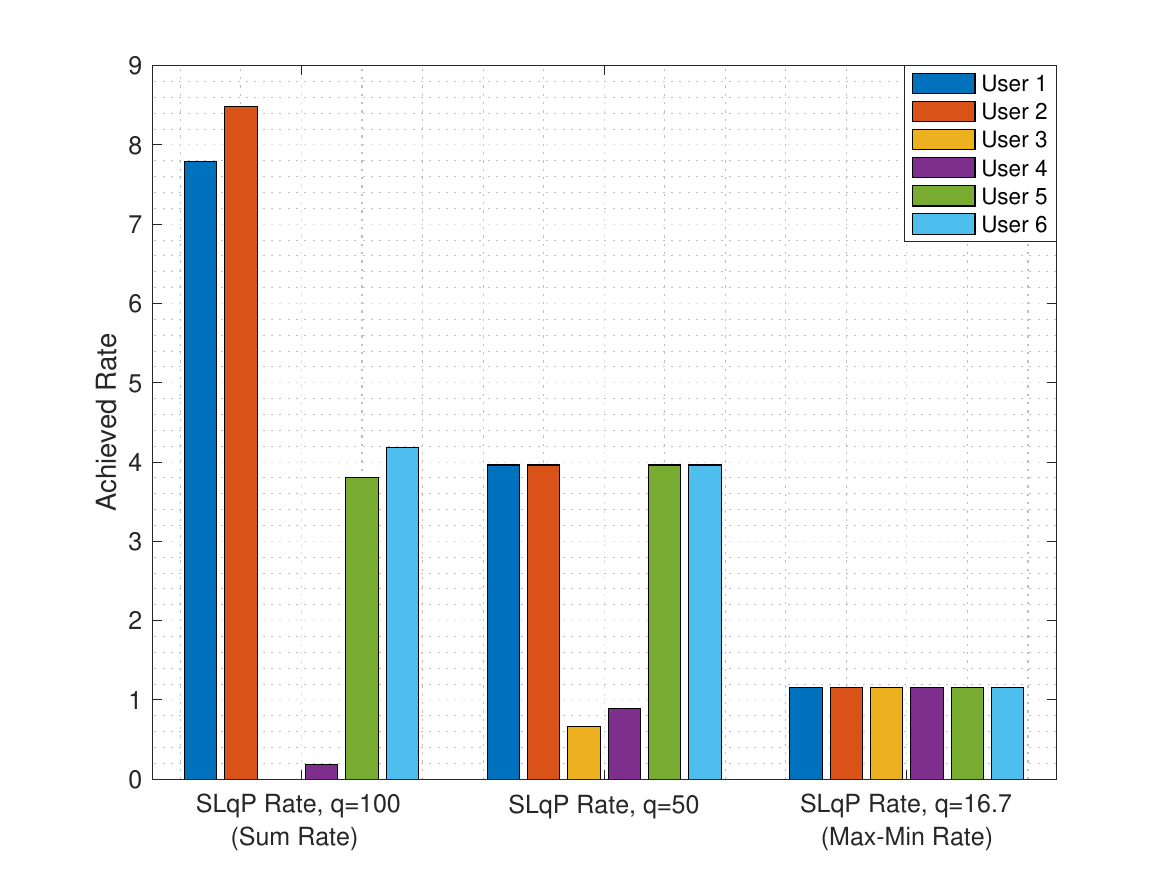}
		\caption{Achieved rates for optimal solutions to the parallel channel SLqP rate optimization problem for different values of $q$.}
		\centering
		\label{interference_free_results}
	\end{center}
\end{figure}
{\color{black}Figure \ref{interference_free_results} illustrates the rates achieved by users corresponding to the optimal solution of Problem (\ref{SPR_problem_interference_free}), obtained using the CVX convex solver, for $q=100$ (sum-rate), $50$ ($50^\mathrm{th}$ percentile) and $16.\overline{6}$ (minimum rate) respectively.} In particular, the $50^\mathrm{th}$ percentile setting maximizes the rates of the worst half (i.e., 3 out of 6) users. We  observe that the achieved rates vary significantly depending upon the percentile we choose to optimize. For the max-sum-rate problem, Users 1 and 2 achieve extremely high rates but this comes at the expense of zero and near-zero rates achieved by User 3 and User 4 respectively due to their high receiver noise powers. In contrast, for the max-min-rate problem, as expected from Property \ref{largest_K_minus_Kq_rates_equal}, all users achieve equal rates. This comes at the cost of considerably reduced rates for Users 1 and 2.

On the other hand, the $50^\mathrm{th}$ percentile SLqP rate optimization problem effects a favorable compromise: Users 3 and 4 achieve significantly improved rates compared to the max-sum-rate setting while the rates achieved by Users 1 and 2 are not curtailed as drastically as in the max-min-rate formulation. In moving from the $50^\mathrm{th}$ percentile problem to the max-min-rate problem, the rates achieved by Users 3 and 4 improve by only $73.6\%$ and $26.8\%$ respectively yet this is achieved at the expense of a $70.8\%$ reduction in the rates for Users 1, 2, 5, and 6.

These results illustrate a fundamental property of the SLqP rate utility: higher values of $q$ favor users with better transmission conditions; conversely, lower values of $q$ favor weaker users. Consequently, the value of $q$ allows us to control the tradeoff between allocating resources to either the weakest or strongest users in the network. The max-sum-rate and max-min-rate optimization problems are clearly the extreme endpoints of the class of SLqP rate optimization problems, and as such, neither may be desirable in terms of the compromises they give rise to by either favoring the strongest or weakest users in the network.

We also note that for the given example, we have $\ensuremath{f_{6,\mathrm{opt}}}=24.4$, $\ensuremath{f_{3,\mathrm{opt}}}=5.52$ and $\ensuremath{f_{1,\mathrm{opt}}}=1.15$; the ordering is as expected from Remark \ref{ordering_by_percentile}. The property from Remark \ref{largest_K_minus_Kq_rates_equal} also holds; as expected, for $q=50$, the $(K-K_q+1=4)$ largest rates are identical.

{\color{black}{These results also illustrate that there is no inherent \textit{optimum} choice of percentile. Ultimately, the choice of this parameter is up to the network operators; for instance, if they wish to maximize sum-rate (and favour the cell-centre users), they would set $q=100$. As discussed earlier, however, the main focus of our work is to help improve cell-edge service to enable new applications \cite{3gpp} which is emphasized by selecting lower percentiles (e.g., $q=5$ or $q=10$).
}}

\subsection{Least-$q^\mathrm{th}$-Percentile (LqP) Rate Optimization}
At this juncture, it is natural to question why we would choose to optimize the sum of the smallest $K_q$ rates rather than the $K_{q}^\mathrm{th}$ smallest rate (which we shall henceforth refer to as the least-$q^\mathrm{th}$-percentile (LqP) rate). It turns out that the optimal solution to the LqP rate maximization problem has undesirable properties which makes it unsuitable for use in the context of improving service for cell-edge users.

Observing that the $K_{q}^\mathrm{th}$ element in a vector can be written as the difference between the sum of the smallest $K_q$ elements and the sum of the smallest $K_q-1$ elements, it follows that the LqP rate maximization problem is given by
\begin{subequations}\label{PPR_problem_interference_free}
	\begin{align}
		\begin{split}
			\underset{\mathbf{p}}{\mathrm{maximize}}\quad f_{K_{q}}\left(r_{1}\left(\mathbf{p}\right),\ldots,r_{K}\left(\mathbf{p}\right)\right)\hspace{7.00em}\\
			-f_{K_{q}-1}\left(r_{1}\left(\mathbf{p}\right),\ldots,r_{K}\left(\mathbf{p}\right)\right)\label{LqP_rate}\hspace{3.00em}
		\end{split}\\
		\begin{split}
			\mathrm{subject\,to}\quad p_{k}\geq0,\,k=1,\ldots,K\label{PPR_constraint_nonnegative}\hspace{0.30em}
		\end{split}\\
		\begin{split}
			\hspace{5.2em}\sum_{k=1}^{K}p_{k}\leq{P}_\mathrm{max}\label{PPR_constraint_total_power}
		\end{split}
	\end{align}
\end{subequations}

\begin{remark} \label{PPR_zero_rates}
	Denote by $\ensuremath{r_{k,\mathrm{opt}}^{}}$ the rate achieved by the $k^\mathrm{th}$ user in the optimal solution for the problem in (\ref{PPR_problem_interference_free}). Then we must have $\ensuremath{\ensuremath{r_{1,opt}^{}=r_{2,opt}^{}=,\ldots,=r_{K_{q}-1,opt}^{}=0}}$.
\end{remark}
\begin{proof}
	Suppose that we are given an optimal solution to Problem (\ref{PPR_problem_interference_free}) in which $r_{k,opt}^{}\neq{0}$ for some user ${k\in\left\{ 1,2,\ldots,K_{q}-1\right\}}$ (and hence the power for this user, $p_k$, is non-zero). Since the LqP rate in (\ref{LqP_rate}), corresponding to the rate achieved by user $K_q$, does not depend on the rates achieved by users $K_q-1,\ldots,1$, it can always be improved by increasing $p_{K_q}$ by $p_k$ and setting user $k$'s power to zero. Note that this change ensures that the sum-power constraint in (\ref{PPR_constraint_total_power}) is respected. Thus, we conclude that the optimal solution to Problem (\ref{PPR_problem_interference_free}) would assign zero powers to the $K_q-1$ weakest users.
\end{proof}

We note that while the LqP utility function is non-concave for $K_q\neq{K}$ (since it is the difference of two concave functions), Remark \ref{PPR_zero_rates} allows us to straightforwardly derive the optimal solution for Problem (\ref{PPR_problem_interference_free}).

\begin{remark} \label{PPR_optimal_solution}
	The optimal solution to Problem (\ref{PPR_problem_interference_free}) can be found by solving the following max-min-rate problem:
	\begin{subequations}\label{PPR_problem_equiv_max_min}
		\begin{align}
			\underset{\mathbf{p}}{\mathrm{maximize}}\quad f_{1}\left(r_{K_q}\left(\mathbf{p}\right),\ldots,r_{K}\left(\mathbf{p}\right)\right)\label{PPR_obj_interference_free_max_min}\hspace{5.60em}\\
			\mathrm{subject\,to}\quad p_{k}\geq0,\,k=1,\ldots,K\label{PPR_constraint_nonnegative}\hspace{6.9em}\\
			\sum_{k=1}^{K}p_{k}\leq {P}_\mathrm{max}\label{PPR_constraint_total_power_max_min}\hspace{9.9em}
		\end{align}
	\end{subequations}
\end{remark}
\begin{proof}
	From Remark \ref{PPR_zero_rates}, it follows that since the rates achieved by users $1,\ldots,K_q-1$ are zero in the optimal solution, the LqP objective in (\ref{LqP_rate}) is equivalent to the smallest rate achieved among the rates achieved among the remaining users $K_q,K_q+1,\ldots,K$.
\end{proof} 

This assignment of zero rates to $K_q-1$ users makes the maximization of LqP rate utility undesirable. This result also holds true in the wireless cellular networks we consider in Section \ref{interference_channel}, since transmitting to users below the $q^\mathrm{th}$-percentile would lead to interference, thereby decreasing the LqP objective function value. We henceforth focus on SLqP-rate problems.

\section{Short-Term Sum-Percentile Rate Optimization}\label{interference_channel}
\subsection{Problem Formulation}
{\color{black}We are now ready to tackle the power control problem for SLqP rate optimization in the more realistic setting of an interference-limited network which models real-world cellular networks. Unlike the parallel Gaussian channel setting, the problem is now no longer convex, necessitating the development of a distinct algorithmic framework to solve it effectively.}

We consider a network with $K$ interfering single-antenna users with the channel from the transmitter of user $j$ to the receiver of user $k$ denoted by $h_{j\rightarrow k}^{}$. Furthermore, the transmit power and AWGN receiver noise power for user $k$ are denoted by $p_{k}^{}$ and ${\sigma}^{2}$ respectively. It follows that the rate achievable by user $k$ is given by
\begin{equation}\label{rates_expression}
r_{k}^{}\left(\mathbf{p}^{}\right)=\mathrm{log}\left(1+\frac{p_{k}^{}\left|h_{k\rightarrow k}^{}\right|^{2}}{\sum_{j\neq k}p_{j}^{}\left|h_{j\rightarrow k}^{}\right|^{2}+{\sigma}^{2}}\right)
\end{equation}

Accordingly, the SLqP rate optimization problem for the $n^\mathrm{th}$ time slot can be expressed as:
\begin{subequations}\label{SPR_problem_shortterm}
	\begin{align}
\underset{\mathbf{p^{\mathit{}}}}{\mathrm{maximize}}\quad f_{K_{q}}\left(r_{1}^{}\left(\mathbf{p}^{}\right),\ldots,r_{K}^{}\left(\mathbf{p}^{}\right)\right)\hspace{2.55em}\label{SPR_obj_shortterm}\\
\mathrm{subject\,to}\quad0\leq p_{k}^{}\leq P_{\mathrm{max}},\,k=1,\ldots,K\label{SPR_constraint_shortterm}\hspace{0.30em}
	\end{align}
\end{subequations}
where $P_T$ denotes the maximum per-link power constraint.

{\color{black}As with the parallel Gaussian setting, this problem formulation allows us to flexibly control the tradeoff between cell-centre and cell-edge users in wireless communication networks. In particular, to target the cell-edge users, we can choose smaller values of $q$. For example, with $q=5$ the problem in (\ref{SPR_problem_shortterm}) optimizes the $5^\mathrm{th}$-percentile throughput in a wireless cellular network. Therefore, this formulation achieves the practical goal of improving data rates for cell-edge users in accordance with the 3GPP and other targets for 6G networks \cite{3gpp,5gppp,ziegler_6g_2020}.}

Although similar in appearence to Problem (\ref{SPR_problem_interference_free}), this problem is considerably more challenging to solve.

\begin{theorem}\label{NP_hard_theorem}
	For a fixed value of $q$ such that $K_q>1$, the percentile program in (\ref{SPR_problem_shortterm}) is strongly NP-hard in the number of users $K$.
\end{theorem}

The proof, detailed in Appendix \ref{NP_hardness_proof}, is based on a polynomial-time reduction from the maximum independent set problem. This reduction is considerably more challenging than the sum-rate setting since we have to consider \textit{the entire class} of SLqP rate maximization problems which, in general, can have a countably infinite number of instances as the number of users grows. 

Further, we also remark that the proof of NP-hardness holds only for non-minimum percentiles, i.e., $K_q>1$ (corresponding to $q>100/K$). This is in agreement with the prior results in the literature: in contrast to the max-sum-rate problem, the max-min-rate power control problem can be expressed as a parametric linear program and solved to optimality in polynomial time \cite{razaviyayn_linear_2013}.

The NP-hardness of Problem (\ref{SPR_problem_shortterm}) has a major implication for solution techniques: finding the global optimum is computationally intractable; thus, reaching directional stationary points is the best we can hope to achieve with polynomial-time (i.e., computationally efficient) algorithms \cite{zhi-quan_luo_dynamic_2008}.

\begin{remark}
	Problem (\ref{SPR_problem_shortterm}) is non-convex and non-smooth for $K_{q}\notin\left\{ 1,K\right\}$.
\end{remark}
To establish the non-convexity, we note that the objective in (\ref{SPR_obj_shortterm}) can be expressed as the pointwise maximum of the sum of every subset of $K_q$ users in the network. The sum-of-rates function (for $K_q>1$) is well-known to be non-convex in the powers \cite{shen_fractional_2018-1,shi_iteratively_2011}; hence, the pointwise maximum of these sums is also non-convex in general.

The non-smoothness follows directly from Property \ref{nonsmoothness_SGqP_SLqP} and has important consequences for the choice of optimization algorithm for solving Problem (\ref{SPR_problem_shortterm}): first-order methods like steepest descent cannot be directly implemented \cite{burke_robust_2005}, while second-order algorithms like Newton's method are similarly inapplicable since the Hessian is not defined \cite{boyd_convex_2004}.

In order to deal with non-smooth problems (such as max-min-rate), the min-max WMSE \cite{razaviyayn_linear_2013}, FP \cite{shen_fractional_2018}, and numerous prior \cite{scutari_parallel_2017} works all introduce additional variables and constraints to `smooth' the problem. Introducing a slack variable $t$ and applying this approach to Problem (\ref{SPR_problem_shortterm}), we obtain the following equivalent smooth problem:
\begin{subequations}\label{SPR_problem_smooth}
	\begin{align}
	\underset{t,\hspace{0.10em}\mathbf{p^{\mathit{}}}}{\mathrm{maximize}}\quad t\hspace{2em}\label{SPR_obj_smooth}\hspace{12.30em}\\
\mathrm{subject\,to}\quad\sum_{l=1}^{K_{q}}r_{i_{l}}^{}\left(\mathbf{p}^{}\right)\geq t,\,\begin{array}{c}
i_{k}\in\left\{ 1,\ldots,K\right\} \\
i_{j}\neq{i}_{k},\forall{j\neq{k}}
\end{array}\label{SPR_constraint_smooth}\\
0\leq p_{k}^{}\leq P_{\mathrm{max}},\,k=1,\ldots,K\label{SPR_constraint_smooth_power}\hspace{1.30em}
	\end{align}
\end{subequations}

This problem is now amenable to gradient-based techniques. Yet this smoothed reformulation suffers from a serious issue which makes its optimization computationally intractable. To illustrate this, we consider the following example. Suppose that we have a SISO network with 4 users, and wish to optimize the $50^\mathrm{th}$-percentile throughput. This means that $K_q=0.50\times4=2$, i.e., we are maximizing the sum of the smallest two rates.

\begin{subequations}
	\begin{align}
		\underset{\mathbf{p},t}{\mathrm{maximize}}\quad t\hspace{12.40em}\label{smoothed_SLqP}\\
		\mathrm{subject\,to}\quad0\leq p_{k}^{}\leq P_{\mathrm{max}},\,k=1,\ldots,4\label{}\hspace{0.60em}\\
		r_{1}\left(\mathbf{p}\right)+r_{2}\left(\mathbf{p}\right)\geq t\label{c11}\hspace{5.00em}\\		r_{1}\left(\mathbf{p}\right)+r_{3}\left(\mathbf{p}\right)\geq t\label{c21}\hspace{5.00em}\\	r_{1}\left(\mathbf{p}\right)+r_{4}\left(\mathbf{p}\right)\geq t\label{c31}\hspace{5.00em}\\		r_{2}\left(\mathbf{p}\right)+r_{3}\left(\mathbf{p}\right)\geq t\label{c41}\hspace{5.00em}\\		r_{2}\left(\mathbf{p}\right)+r_{4}\left(\mathbf{p}\right)\geq t\label{c51}\hspace{5.00em}\\
		r_{3}\left(\mathbf{p}\right)+r_{4}\left(\mathbf{p}\right)\geq t\label{c61}\hspace{5.00em}
	\end{align}
\end{subequations}

In other words, we wish to maximize the slack variable $t$ such that the sum of \textit{any two rates} exceeds it. There are 
\[
_{4}C_{2}=\frac{4!}{\left(4-2\right)!2!}=6
\]
ways of choosing 2 users out of 4; hence, the number of additional constraints introduced is 6.

In general, for a network with $K$ users and optimizing the throughput for the weakest $K_q$ users, the smoothed variant of the SLqP rate maximization problem will comprise \[_{K}C_{K_q}=\frac{K!}{\left(K-K_q\right)!K_q!}\] constraints. 

The issue now becomes obvious: even for very small problem sizes, dealing with this combinatorial number of constraints quickly becomes impractical. To make matters worse, the constraints (e.g. (\ref{c11}-\ref{c61})) are in a \textit{sum-of-rates} form, which is known to be non-convex \cite{shi_iteratively_2011,shen_fractional_2018,shen_fractional_2018-1}. Therefore, applying the prior techniques incurs formidable computational complexity.

In summary, published approaches used to optimize classical metrics do not apply to our percentile problem. Therefore, we seek an alternative method which allows us to overcome these formidable challenges.
\vspace{-1.0em}
\subsection{Proposed Approaches}
To derive effective optimization strategies for the problem in (\ref{SPR_problem_shortterm}), we develop techniques to transform it into equivalent block-concave forms. Specifically, we develop two algorithmic approaches that allow for iterative optimization of the original problem to stationarity. We discuss each of these separately in the subsequent subsections.
\vspace{0.5em}
\hrule
\vspace{0.5em}
\subsubsection{\textbf{Quadratic Fractional Transform}}
Our first technique utilizes the quadratic fractional transform (QFT), first proposed in \cite{benson_global_2004} by Benson and subsequently exploited by Shen and Yu in \cite{shen_fractional_2018-1} primarily for weighted sum-rate maximization problems. As explained in the previous section, while this transform has also been utilized for the non-smooth max-min-rate optimization problem \cite{shen_fractional_2018-1}, the direct application of this transform to Problem (\ref{SPR_problem_shortterm}) for non-minimum percentiles would be based on the aforementioned hypograph form with $_{K}C_{K_q}$ constraints. In contrast, we develop an approach that allows us to transform the original problem while side-stepping the associated combinatorial number of non-convex constraints.

We begin by re-stating the transform below:
\begin{lemma}\label{quadratic_transform}
	Let $A\left(\mathbf{p}\right):\mathbb{R}^{K}\mapsto\mathbb{R}_{+}$, $B\left(\mathbf{p}\right):\mathbb{R}^{K}\mapsto\mathbb{R}_{++}$ and $\mathcal{P\subseteq}\mathbb{R}^{K}$. Then the fractional optimization problem
	\begin{subequations}
		\begin{align}
		\underset{\mathbf{p^{\mathit{}}}}{\mathrm{maximize}}\quad\frac{A\left(\mathbf{p}\right)}{B\left(\mathbf{p}\right)}\hspace{5.0em}\label{quad_transform_orig_obj}\\
		\mathrm{subject\,to}\quad\mathbf{p}\in\mathcal{P}\hspace{4.9em}
		\end{align}
	\end{subequations}
is equivalent to the following auxiliary optimization problem
	\begin{subequations}
	\begin{align}
\underset{x,\mathbf{p^{\mathit{}}}}{\mathrm{maximize}}\quad2x\sqrt{A\left(\mathbf{p}\right)}-x^{2}B\left(\mathbf{p}\right)\hspace{0.5em}\label{quadratic_transform_aux_obj}\\
		\mathrm{subject\,to}\quad\mathbf{p}\in\mathcal{P}\hspace{6.8em}
	\end{align}
\end{subequations}
where the equivalence is in the sense of both the optimal variables and objective function value.
\end{lemma}
\begin{proof}
	Observe that the objective in (\ref{quadratic_transform_aux_obj}) is concave in the auxiliary variable $x$. Furthermore, setting the first order condition with respect to this variable, we obtain:
	\[
	\frac{\partial}{\partial x}\left\{ 2x\sqrt{A\left(\mathbf{p}\right)}-x^{2}B\left(\mathbf{p}\right)\right\} =0\Rightarrow x=\frac{\sqrt{A\left(\mathbf{p}\right)}}{B\left(\mathbf{p}\right)}
	\]
	Substituting this value of $x$ back into the objective in (\ref{quadratic_transform_aux_obj}) we obtain the original objective in (\ref{quad_transform_orig_obj}).
\end{proof}
Based on this transform, we now proceed to derive an equivalent reformulation of the problem in (\ref{SPR_problem_shortterm}).

\begin{theorem}\label{quadratic_SPR_equivalence}
The optimization problem in (\ref{SPR_problem_shortterm}) is equivalent to the following auxiliary problem:
\begin{subequations}\label{SPR_problem_quad}
	\begin{align}
	\underset{\mathbf{x}^{},\mathbf{p^{\mathit{}}}}{\mathrm{maximize}}\quad f_{K_{q}}\left(\hat{r}_{1}^{}\left(x_{1}^{},\mathbf{p}^{}\right),\ldots,\hat{r}_{K}^{}\left(x_{K}^{},\mathbf{p}^{}\right)\right)\hspace{0.em}\label{SPR_obj_quad}\\
	\mathrm{subject\,to}\quad0\leq p_{k}^{}\leq P_{\mathrm{max}},\,k=1,\ldots,K\label{SPR_constraint_shortterm_quad}\hspace{1.00em}
	\end{align}
\end{subequations}
where
\begin{multline}
\hat{r}_{k}^{}\left(x_{k}^{},\mathbf{p}^{}\right)=\mathrm{log}\left(1+2x_{k}^{}\sqrt{A_{k}\left(\mathbf{p}^{}\right)}-\right.\left.\left(x_{k}^{}\right)^{2}B_{k}\left(\mathbf{p}^{}\right)\right)\label{auxiliary_rate_QFT}
\end{multline}
and $A_{k}\left(\mathbf{p}^{}\right)$ and $B_{k}\left(\mathbf{p}^{}\right)$ represent the signal and interference-plus-noise powers at the $k^{th}$ user respectively, i.e.,
\begin{subequations}
	\begin{align}
	A_{k}\left(\mathbf{p}^{}\right)=p_{k}^{}\left|h_{k\rightarrow k}^{}\right|^{2}\label{A_value}\hspace{3.82em}\\
	B_{k}\left(\mathbf{p}^{}\right)=\sum_{j\neq k}p_{j}^{}\left|h_{j\rightarrow k}^{}\right|^{2}+\sigma^{2}\label{B_value}
	\end{align}
\end{subequations}
The equivalence between the two problems is in the sense of the optimal variables and objective function values.
\end{theorem}
The details of the proof are provided in Appendix \ref{h_quad_proof}.

The problem in (\ref{SPR_problem_quad}) is now in a form that is amenable to a cyclic optimization strategy. First, we observe that with $\mathbf{p}^{}$ fixed, an optimal choice of $\mathbf{x}^{}$ to maximize the auxiliary objective in (\ref{SPR_obj_quad}) is given by Lemma \ref{quadratic_transform} as
\begin{equation}\label{optimal_quadratic_variables}
x_{k}^{}=\frac{\sqrt{A_{k}\left(\mathbf{p}^{}\right)}}{B_{k}\left(\mathbf{p}^{}\right)},\quad{k}{=}{1,\ldots,K}
\end{equation}
where $A_{k}\left(\mathbf{p}^{}\right)$ and $B_{k}\left(\mathbf{p}^{}\right)$ are given as in (\ref{A_value}) and (\ref{B_value}) respectively. 

\begin{remark}\label{optimal_updates_uniqueness}
It is crucial to note that, unlike other uses of fractional programming in the literature \cite{shen_fractional_2018-1,wei_noma_2019,khan_optimizing_2020}, $f_{K_{q}}\left(\cdot\right)$ is not strictly concave; hence, there may be (infinitely many) other optimal choices of $\mathbf{x}^{}$ that achieve an identical maximum value of the auxiliary objective. However, the update in (\ref{optimal_quadratic_variables}) uniquely preserves the equivalence with the original problem in (\ref{SPR_problem_shortterm}).
\end{remark}

We next turn our attention to optimizing the power variables. To do so, we make use of the following critical insight:
\begin{remark}\label{quad_block_concave}
	The auxiliary objective function in (\ref{SPR_obj_quad}) is concave in the power variables $\mathbf{p}^{}$ when the auxiliary variables $\mathbf{x}^{}$ are fixed.
\end{remark}
This block-concavity can be understood as follows: observe that the terms $A_{k}\left(\mathbf{p}^{}\right)$ and $B_{k}\left(\mathbf{p}^{}\right)$ are linear and affine functions of $\mathbf{p}^{}$ respectively and hence concave. From (\ref{optimal_quadratic_variables}), we note that the auxiliary variables $x_{k}^{}$ are always non-negative owing to the non-negativity of $A_{k}\left(\mathbf{p}^{}\right)$ and $B_{k}\left(\mathbf{p}^{}\right)$. Taken together with the concavity and monotonicity of the logarithm function, it follows that $\tilde{r}_{k}^{}\left(x_{k}^{},\mathbf{p}^{}\right)$ is concave in $\mathbf{p}^{}$. Finally, from Properties \ref{convexity_concavity} and \ref{SLqP_SGqP_nondecreasing}, we note that $f_{K_{q}}\left(\cdot\right)$ is concave and non-decreasing in each component. Therefore, similar to the parallel Gaussian channel setting, it follows from the composition rule that the auxiliary objective in (\ref{SPR_obj_quad}) is concave in the power variables when the auxiliary variables are fixed. Thus, an optimal $\mathbf{p}^{}$ can be obtained straightforwardly when $\mathbf{x}^{}$ is fixed by solving a convex optimization problem using CVX or the subgradient method as discussed earlier.

Combining these updates together, Algorithm 1 summarizes the proposed QFT approach for maximizing the SLqP rate in (\ref{SPR_obj_shortterm}).
\begin{algorithm}
	\caption{QFT Algorithm for SLqP Rate Maximization}
	\label{QFTAlg}
	\begin{algorithmic}[1]
		\State \textbf{initialize} $\mathbf{p}^{}$
		\For{$i=1,\ldots$}
		\State \textbf{update} $\mathbf{x}^{}$ using (\ref{optimal_quadratic_variables}).
		\State \textbf{update} $\mathbf{p}^{}$ by solving (\ref{SPR_problem_quad}) for fixed $\mathbf{x}^{}$.
		\EndFor
		\State \textbf{until} some convergence criterion is met.
	\end{algorithmic}
\end{algorithm}

	\begin{figure*}
	\small
	\begin{subequations}
		\begin{align}
			\hspace{-9em}f_{K_{q}}\left(r_{1}^{}\left(\mathbf{p}^{}\left[i+1\right]\right),\ldots,r_{K}^{}\left(\mathbf{p}^{}\left[i+1\right]\right)\right)=f_{K_{q}}\left(\tilde{r}_{1}^{}\left(x_{1}^{}\left[i+1\right],\mathbf{p}^{}\left[i+1\right]\right),\ldots,\tilde{r}_{K}^{}\left(x_{K}^{}\left[i+1\right],\mathbf{p}^{}\left[i+1\right]\right)\right)\hspace{-6em}\label{convergence_proof_line_1}\\
			\geq f_{K_{q}}\left(\tilde{r}_{1}^{}\left(x_{1}^{}\left[i+1\right],\mathbf{p}^{}\left[i\right]\right),\ldots,\tilde{r}_{K}^{}\left(x_{K}^{}\left[i+1\right],\mathbf{p}^{}\left[i\right]\right)\right)\hspace{-2.45em}\label{convergence_proof_line_2}\\
			\geq f_{K_{q}}\left(\tilde{r}_{1}^{}\left(x_{1}^{}\left[i\right],\mathbf{p}^{}\left[i\right]\right),\ldots,\tilde{r}_{K}^{}\left(x_{K}^{}\left[i\right],\mathbf{p}^{}\left[i\right]\right)\right)\hspace{1.10em}\label{convergence_proof_line_3}\\
			=f_{K_{q}}\left(r_{1}^{}\left(\mathbf{p}^{}\left[i\right]\right),\ldots,r_{K}^{}\left(\mathbf{p}^{}\left[i\right]\right)\right)\hspace{6.95em}\label{convergence_proof_line_4}
		\end{align}
	\end{subequations}
	\hrule{}
\end{figure*}
\begin{theorem}\label{QFT_nondecreasing}
Algorithm \ref{QFTAlg} is non-decreasing in the SLqP rate objective given in (\ref{SPR_obj_shortterm}) after each iteration.
\end{theorem}
\begin{proof}
The non-decreasing nature of the algorithm can be understood by considering the chain of reasoning in (\ref{convergence_proof_line_1})--(\ref{convergence_proof_line_4}), where $\mathbf{p}\left[i\right]$ and $\mathbf{x}\left[i\right]$ indicate the values of these variables in the $i^\mathrm{th}$ iteration. Specifically, (\ref{convergence_proof_line_1}) follows from Theorem \ref{quadratic_SPR_equivalence}; (\ref{convergence_proof_line_2}) follows from the fact that updating $\mathbf{p}^{}$ when $\mathbf{x}^{}$ is fixed maximizes the auxiliary objective in (\ref{SPR_obj_quad}); (\ref{convergence_proof_line_3}) follows from similar reasoning applied to the update of $\mathbf{x}^{}$ when $\mathbf{p}^{}$ is fixed, and (\ref{convergence_proof_line_4}) follows from Theorem \ref{quadratic_SPR_equivalence}. Further, we note that since the SLqP rate objective is upper bounded, the algorithm must converge.
\end{proof}

\begin{theorem}\label{QFT_convergence}
Algorithm \ref{QFTAlg} is a cyclic minorization-maximization (MM) algorithm and is guaranteed to converge to a directional stationary point of Problem (\ref{SPR_problem_shortterm}).
\end{theorem}
The proof is detailed in Appendix \ref{MM_proof}. We remark that this convergence result is different from those in the prior FP/WMMSE works \cite{shen_fractional_2018,shen_fractional_2018-1,shi_iteratively_2011} which primarily focus on reaching stationary points of weighted sum-rate problems. For the WSR (and by extension the special sum-rate case), it is sufficient to show that the objective function is non-decreasing after each iteration of the algorithms. In contrast, the SLqP rate objective function is non-smooth in general, which means that we instead have to prove convergence to a \textit{directional} stationary point, which is considerably more challenging.
\vspace{0.5em}
\hrule
\vspace{0.50em}
\subsubsection{\textbf{Logarithmic Fractional Transform}}

The quadratic transform is well-established in the literature as a means of tackling sum-of-ratios optimization problems \cite{shen_fractional_2018-1,khan_optimizing_2018}. In principle, however, \textit{any} similar transform that converts the rate expression into block-concave form would allow us to optimize SLqP utility in an identical fashion. Hence, the quadratic fractional transform is sufficient, but \textit{not necessary} to solve percentile programs. Accordingly, in this section, we introduce a novel transform, called the logarithmic fractional transform (LFT), which similarly allows us to tackle percentile rate programs by converting them into block-concave forms. While similar in spirit to the QFT, the resultant approach is algorithmically distinct, and thus generally converges to distinct directional stationary points.

\begin{theorem}\label{logarithmic_transform}
	Let $A\left(\mathbf{p}\right):\mathbb{R}^{K}\mapsto\mathbb{R}_{+}$, $B\left(\mathbf{p}\right):\mathbb{R}^{K}\mapsto\mathbb{R}_{++}$ and $\mathcal{P\subseteq}\mathbb{R}^{K}$. Then the fractional optimization problem
	\begin{subequations}
		\begin{align}
		\underset{\mathbf{p^{\mathit{}}}}{\mathrm{maximize}}\quad\mathrm{log}\left(1+\frac{A\left(\mathbf{p}\right)}{B\left(\mathbf{p}\right)}\right)\hspace{0.5em}\label{logarithmic_transform_orig_obj}\\
		\mathrm{subject\,to}\quad\mathbf{p}\in\mathcal{P}\hspace{4.9em}
		\end{align}
	\end{subequations}
	is equivalent to the following auxiliary optimization problem
	\begin{subequations}
		\begin{align}
\underset{x,\mathbf{p^{\mathit{}}}}{\mathrm{maximize}}\quad-xB\left(\mathbf{p}\right)+\mathrm{log}\left(x\left(A\left(\mathbf{p}\right)+B\left(\mathbf{p}\right)\right)\right)+1\hspace{0em}\label{logarithmic_transform_aux_obj}\\
		\mathrm{subject\,to}\quad\mathbf{p}\in\mathcal{P}\hspace{14.10em}
		\end{align}
	\end{subequations}
	where the equivalence is in the sense of both the optimal variables and objective function value.
\end{theorem}
\begin{proof}
	The proof is similar to that for Theorem \ref{quadratic_transform}. The objective in (\ref{logarithmic_transform_aux_obj}) is concave in the auxiliary variable $x$; hence, setting the first order condition with respect to this variable, we obtain:
	\begin{multline}
\frac{\partial}{\partial x}\left\{ -xB\left(\mathbf{p}\right)+\mathrm{log}\left(x\left(A\left(\mathbf{p}\right)+B\left(\mathbf{p}\right)\right)\right)+1\right\} =0\\\Rightarrow x=\frac{1}{B\left(\mathbf{p}\right)}
	\end{multline}
	Substituting this value of $x$ back into the objective in (\ref{logarithmic_transform_aux_obj}) we obtain the original objective in (\ref{logarithmic_transform_orig_obj}).
\end{proof}

\begin{theorem}\label{logarithmic_SPR_equivalence}
	The optimization problem in (\ref{SPR_problem_shortterm}) is equivalent to the following auxiliary problem:
	\begin{subequations}\label{SPR_problem_log}
		\begin{align}
		\underset{\mathbf{x}^{},\mathbf{p^{\mathit{}}}}{\mathrm{maximize}}\quad f_{K_{q}}\left(\check{r}_{1}^{}\left(x_{1}^{},\mathbf{p}^{}\right),\ldots,\check{r}_{K}^{}\left(x_{K}^{},\mathbf{p}^{}\right)\right)\hspace{0.em}\label{SPR_obj_log}\\
		\mathrm{subject\,to}\quad0\leq p_{k}^{}\leq P_{\mathrm{max}},\,k=1,\ldots,K\label{SPR_constraint_shortterm_log}\hspace{1.0em}
		\end{align}
	\end{subequations}
	where
	\begin{multline}
	\check{r}_{k}^{}\left(x_{k}^{},\mathbf{p}^{}\right)=-x_{k}^{}B_{k}\left(\mathbf{p}^{}\right)+
	\mathrm{log}\left(x\left(A_{k}\left(\mathbf{p}^{}\right)+B_{k}\left(\mathbf{p}^{}\right)\right)\right)+1
	\end{multline}
	and $A_{k}\left(\mathbf{p}^{}\right)$ and $B_{k}\left(\mathbf{p}^{}\right)$ are the received signal and interference-plus-noise powers at the $k^{th}$ user given in (\ref{A_value}) and (\ref{B_value}) respectively.
	The equivalence between the two problems is in the sense of the optimal variables and objective function values.
\end{theorem}
\begin{proof}
	The proof follows from a reasoning similar to Theorem \ref{quadratic_SPR_equivalence}; the details are omitted for brevity.
\end{proof}
Similar to the quadratic transform auxiliary objective, the logarithmic transform auxiliary objective in (\ref{SPR_problem_log}) is now in a form amenable to optimization. With the power variables $\mathbf{p}^{}$ fixed, an optimal choice for the auxiliary variables $\mathbf{x}^{}$ is given by
\begin{equation}\label{optimal_logarithmic_variables}
x_{k}^{}=\frac{1}{B_{k}\left(\mathbf{p}^{}\right)},\quad{k}=1,\ldots,K
\end{equation}
\begin{remark}\label{log_auxiliary_concave}
	The auxiliary objective function in (\ref{SPR_obj_log}) is concave in the power variables $\mathbf{p}^{}$ when the auxiliary variables $\mathbf{x}^{}$ are fixed.
\end{remark}
\begin{proof}
	The proof follows from similar reasoning to Remark \ref{quad_block_concave}.
\end{proof}

From the preceding remark, it follows that a similar block coordinate ascent strategy to Algorithm \ref{QFTAlg} can be utilized to derive an effective optimization strategy by alternately fixing the power variables $\mathbf{p}^{(n)}$ while updating $\mathbf{x}^{(n)}$, and vice versa. Combining these update steps together, we summarize the second proposed approach for SLqP rate optimization in Algorithm \ref{LFTAlg}.

\begin{algorithm}
	\caption{LFT Algorithm for SLqP Rate Maximization}
	\label{LFTAlg}
	\begin{algorithmic}[1]
		\State \textbf{initialize} $\mathbf{p}^{}$
		\For{$i=1,\ldots$}
		\State \textbf{update} $\mathbf{x}^{}$ using (\ref{optimal_logarithmic_variables}).
		\State \textbf{update} $\mathbf{p}^{}$ by solving (\ref{SPR_problem_log}) for fixed $\mathbf{x}^{}$.
		\EndFor
		\State \textbf{until} some convergence criterion is met.
	\end{algorithmic}
\end{algorithm}
\newpage
\begin{theorem} \label{LFT_convergence}
	Algorithm \ref{LFTAlg} is a cyclic MM algorithm and converges in a non-decreasing fashion to reach a directional stationary point of the original SLqP rate objective in (\ref{SPR_obj_shortterm}).
\end{theorem}
\begin{proof}
	The proof follows from similar reasoning to Theorems \ref{QFT_nondecreasing} and \ref{QFT_convergence}; the details are omitted for brevity.
\end{proof}

\subsection{Numerical Results}\label{intf_results}
{\color{black}
To evaluate the performance of the proposed algorithms, we simulate a network with $7$ hexagonal wrapped-around cells and BSs located at the center of each cell. Each cell contains $8$ users; thus, the total number of users in the network is $K=56$. The distance between adjacent BSs is set as $2000\hspace{0.1em}\mathrm{m}$, and we assume Rayleigh fading with a block-fading model. The path loss between a user and BS separated by a distance of $d$ meters is given by $\left(1+{d}/{d_{0}}\right)^{-\zeta/2}$ where $d_0=0.3920$ is a model-dependent reference distance, and $\zeta=3.76$ is the pathloss exponent. The noise power spectral density (PSD) is $-143\hspace{0.1em}\mathrm{dBm/Hz}$; we assume a system bandwidth of $20\hspace{0.1em}\mathrm{MHz}$ and maximum per-user transmit power constraint of $43\hspace{0.1em}\mathrm{dBm}$. The path loss exponent and reference distance that we utilize are based on the COST231 model which simulates LTE channels in real-world conditions. As such, these parameters have been identically utilized in numerous prior works including, but not limited to, \cite{shen_fractional_2018,shen_fractional_2018-1,khan_optimizing_2020,hosseini_optimizing_2018}.

We begin by considering the performance of the proposed approach for the SLqP rate maximization problem for $K_q=14$, i.e., the $25^\mathrm{th}$ percentile throughput for a single channel realization with the user and BS locations illustrated in Figure \ref{user_locations}. For these and subsequent results, we initialize each user's transmission power by sampling uniformly randomly in the interval $\left[0,P_\mathrm{max}\right]$. We observe from the results in Figure \ref{K8_Kq14_st} that for both the quadratic fractional transform algorithm and logarithmic fractional transform algorithm, the auxiliary objective function values are equivalent after each iteration as per Theorems \ref{quadratic_SPR_equivalence} and \ref{logarithmic_SPR_equivalence}. Furthermore, as predicted by Theorems \ref{QFT_convergence} and \ref{LFT_convergence}, both algorithms converge in a non-decreasing fashion (starting from an identical random initialization) to a stationary point.
\begin{figure}[]
	\begin{center} 
		\includegraphics[trim={0cm 0cm 0cm 0cm},clip,width=0.4\textwidth]{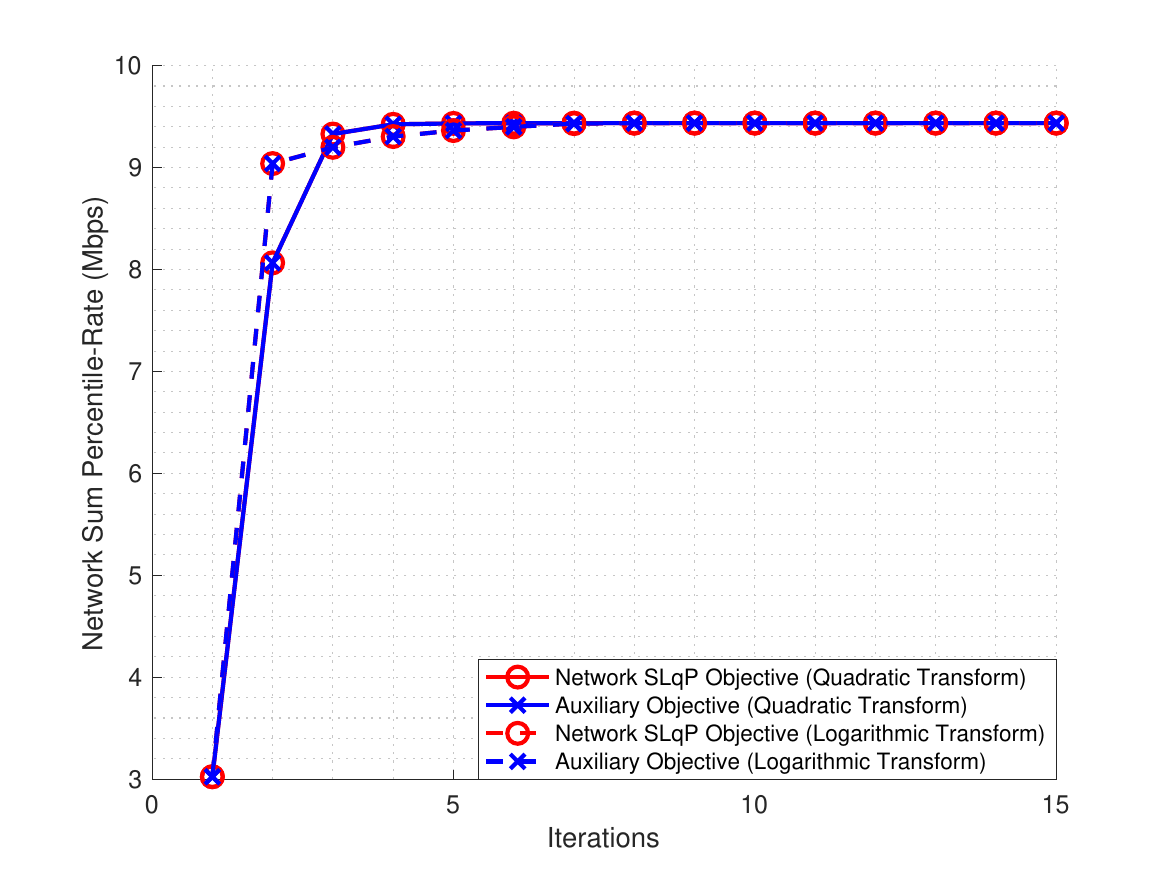}
		\caption{{\color{black}Convergence of proposed algorithms for $25^\mathrm{th}$-percentile SLqP rate maximization problem, $K=56$, $K_q=14$.}}
		\centering
		\label{K8_Kq14_st}
	\end{center}
\end{figure} 

\begin{figure}[]
	\begin{center} 
		\includegraphics[trim={3cm 0cm 3cm 0cm},clip,width=0.4\textwidth]{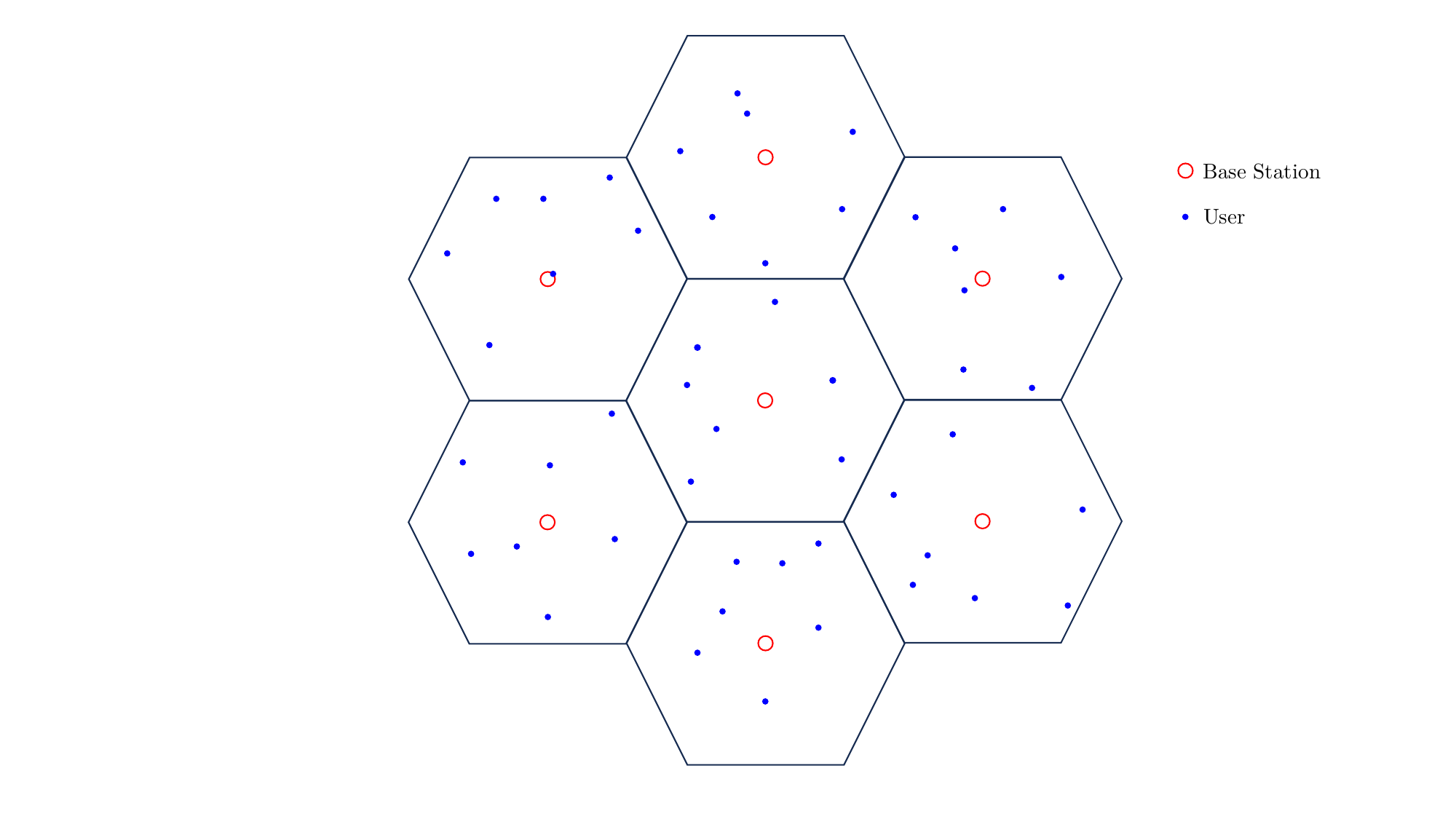}
		\caption{User and BS locations for a single realization in a network with $K=70$ users.}
		\centering
		\label{user_locations}
	\end{center}
\end{figure} 

Additionally, we compare the performance of the proposed QFT and LFT algorithms against several well-known prior algorithms from the resource allocation literature. Specifically, we compare against the following four algorithms:
\begin{itemize}
	\item Successive Convex Approximation (SCA): this approach has widely been used in a number of signal processing problems and resource management algorithms \cite{shen_multi-uav_2020,ibrahim_fast_2020,dong_multi-group_2020}; we use the variant in \cite{scutari_parallel_2017,scutari_parallel_20171}.
	\item Sequential Quadratic Programming (SQP): This has been widely utilized as a benchmark in prior resource management algorithms \cite{khan_optimizing_2020}. For our implementation, we utilize this method using the standard \verb#fmincon# solver in MATLAB.
	\item Subgradient Ascent (SGA): The subgradient ascent method is a heuristic for optimizing non-smooth problems, and is an analogue of the well-known gradient ascent method \cite{nesterov_subgradient_2014}. Specifically, for problems in which the objective function is non-differentiable, this method repeatedly computes and then takes steps in the direction of a subgradient of the objective function. Compared to the aforementioned SCA and SQP methods, this approach has the unique benefit of running much faster, as we do not introduce a combinatorial number of constraints.
	\item Channel Weighted Sum-Rate (CWSR): We also compare against a heuristic of weighted sum-rate in which the weights are set to the inverse of the channel strength for each user. This favours the cell-edge users by weighting them higher as compared to the cell-centre users.
\end{itemize}

Figure \ref{convergence_benchmarks} shows the convergence for a randomly-generated set of identical channel realizations for $K=14$ and $K_q=7$. We observe that the QFT and LFT algorithms converge quickly, and in a non-decreasing fashion, to a directional stationary point within as little as 6 iterations and achieve the highest objective function value. Among the competitors, the SCA approach performs the best and converges in a non-decreasing fashion. On the other hand, the CWSR approach increases initially but then decreases after the second iteration. We note that this heuristic optimizes a different, proxy utility function (i.e., WSR) as compared to the true objective function (i.e., SLqP) rate; hence, this result is unsurprising.
\begin{figure}[t!]
	\begin{center} 
		\includegraphics[trim={0cm 0cm 0cm 0cm},clip, width=0.4\textwidth]{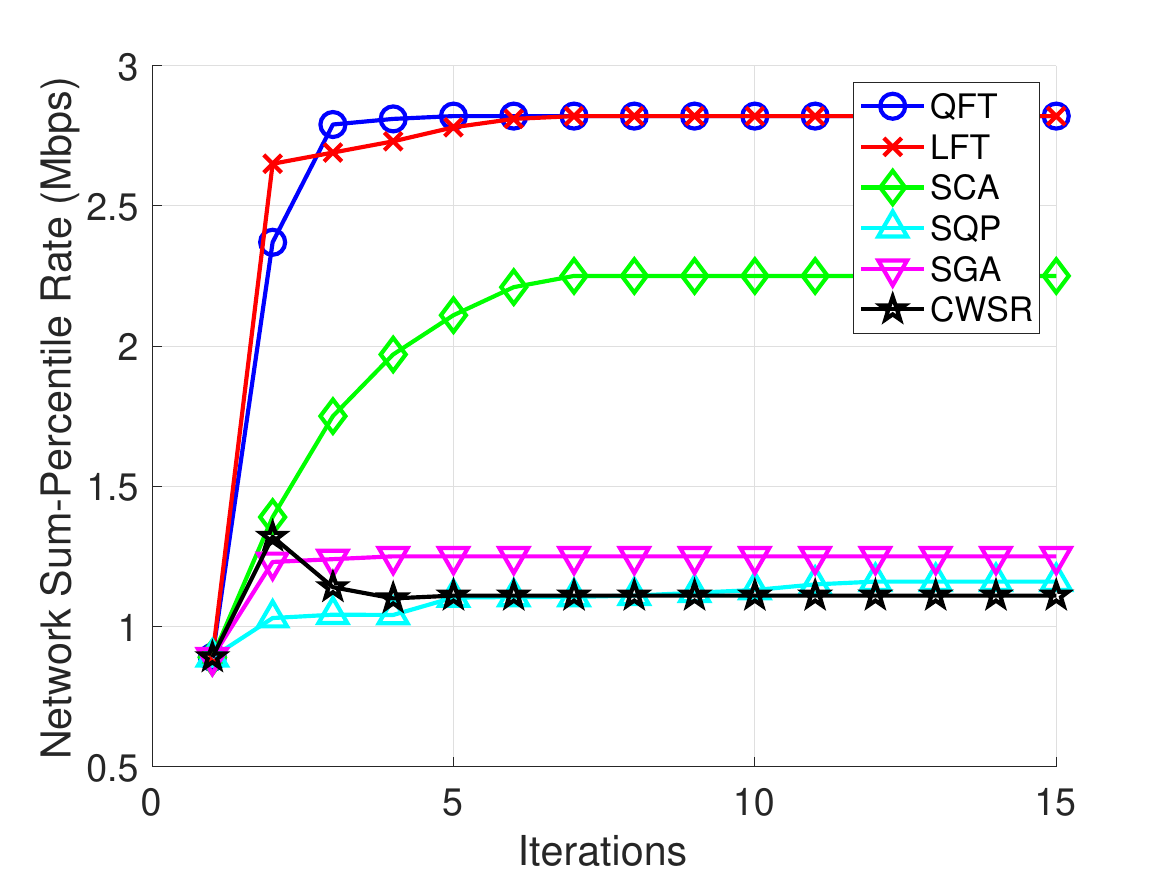}
		\caption{Convergence of SLqP utility for $K=14$ and $K_q=7$.}
		\centering
		\label{convergence_benchmarks}
	\end{center}
\end{figure}

We also considered another set of parameters, with $K=21$ and $K_q=3$ and averaged across a large set of random channel realizations. The results for this case are illustrated in Table \ref{execution_times}. It can be clearly seen that the QFT and LFT algorithms achieve the best performance, followed by the SCA and SGA algorithms. The SCA algorithm has the best performance once again among the competitors but the proposed QFT and LFT algorithms still achieve about 30\% higher objective function values.

As mentioned earlier, an additional critical point that cannot be ignored is the computation time. Since the prior optimization algorithms smooth the problem prior to solving it, the computation time needed to deal with the combinatorial number of non-convex constraints is excessive. To measure the execution times, we ran all algorithms on a machine with a 12-core AMD Ryzen 9 5900X CPU and 128 GB of memory; these are also given in Table \ref{execution_times}. We observe that the proposed QFT and LFT schemes converge to a solution {\textit{{more than 1,000 times faster}}} than the SCA and SQP algorithms. This is to be expected, since the smoothed variant of the problem required by these algorithms will have $_{21}C_{3}=1,330$ non-convex constraints, which will lead to a longer computational time. On the other hand, while the CWSR and SGA approaches execute much faster compared to CVX, they achieve low objective function values. Consequently, we affirm that our QFT and LFT schemes demonstrate superior speed and performance compared to the benchmarks.

We note that it is possible to write custom solvers (e.g. based on subgradient methods \cite{nesterov_subgradient_2014}) which may be able to solve the problem more efficiently than calling CVX. It is also not clear if closed-form algorithms can be found to solve the general SLqP-rate maximization problem; nonetheless, this would be an extremely interesting direction for future research.

\begin{table}[t!]
	\begin{center} 
		\captionof{table}{Execution times and SLqP utility for different resource algorithms.}\label{execution_times}
		\begin{tabular}{|c|c|c|}
			\hline
			Algorithm & SLqP Utility & Execution Time (seconds) \\ \hline
			QFT       & 1.49 & 13                  \\ \hline
			LFT       & 1.47 & 13                   \\ \hline
			SCA  & 1.20 & 18,953                   \\ \hline			
			SQP  & 0.81 & 21,559                 \\ \hline
			SGA  &  0.93 & 24                    \\ \hline
			CWSR & 0.86 & 3                    \\ \hline
		\end{tabular}
		\centering

	\end{center} 
\end{table}

To further evaluate the performance, we next consider SLqP rate optimization for the scenario when there are $K=70$ users in the network for $K_q=7$ (corresponding to $q=10$, and shown in Figure \ref{power}) and $K_q=4$ (corresponding to $q=5.7$, and shown in Figure \ref{SLqP_}). We ran the QFT and LFT algorithms for 1000 identically chosen random channel and user location initializations and plotted the performance as compared to randomized power control, sum-rate optimized power control (using the WMMSE-SR algorithm \cite{shi_iteratively_2011}) and PF control (using the WMMSE-PF algorithm as in \cite{shi_iteratively_2011}) for $P_\mathrm{max}$ values ranging between 10 dBm and 50 dBm. The PF power control scheme aims to maximize the sum of the logarithm of the achieved rates across the network; for a single time slot, this is a convex optimization problem \cite{zhi-quan_luo_dynamic_2008} and hence the WMMSE-PF algorithm converges to the optimum. 

As we observe in Figure \ref{power}, the sum-rate optimized power control achieves nearly zero cell-edge SLqP rate across the entire range of transmission powers; as mentioned earlier, this is to be expected since this optimization is unfair to cell-edge users. Random power control is better than sum-rate optimized power control while PF power control improves on both. Nevertheless, the proposed QFT and LFT algorithms comprehensively outperform even the PF power control scheme, achieving almost \textit{double} the cell-edge throughput compared to the WMMSE-PF power control scheme at the typical LTE power level of 43 dBm. At a transmission power of 50 dBm, the gap between the proposed approaches and the PF benchmark is even larger. Notably, the QFT algorithm marginally outperforms the LFT algorithm across the range of transmit powers. A similar trend holds for the $5.7^\mathrm{th}$-percentile SLqP utility in Figure \ref{SLqP_}.

\begin{figure}[]
	\begin{center} 
		\includegraphics[trim={0cm 0cm 0cm 0cm},clip,width=0.4\textwidth]{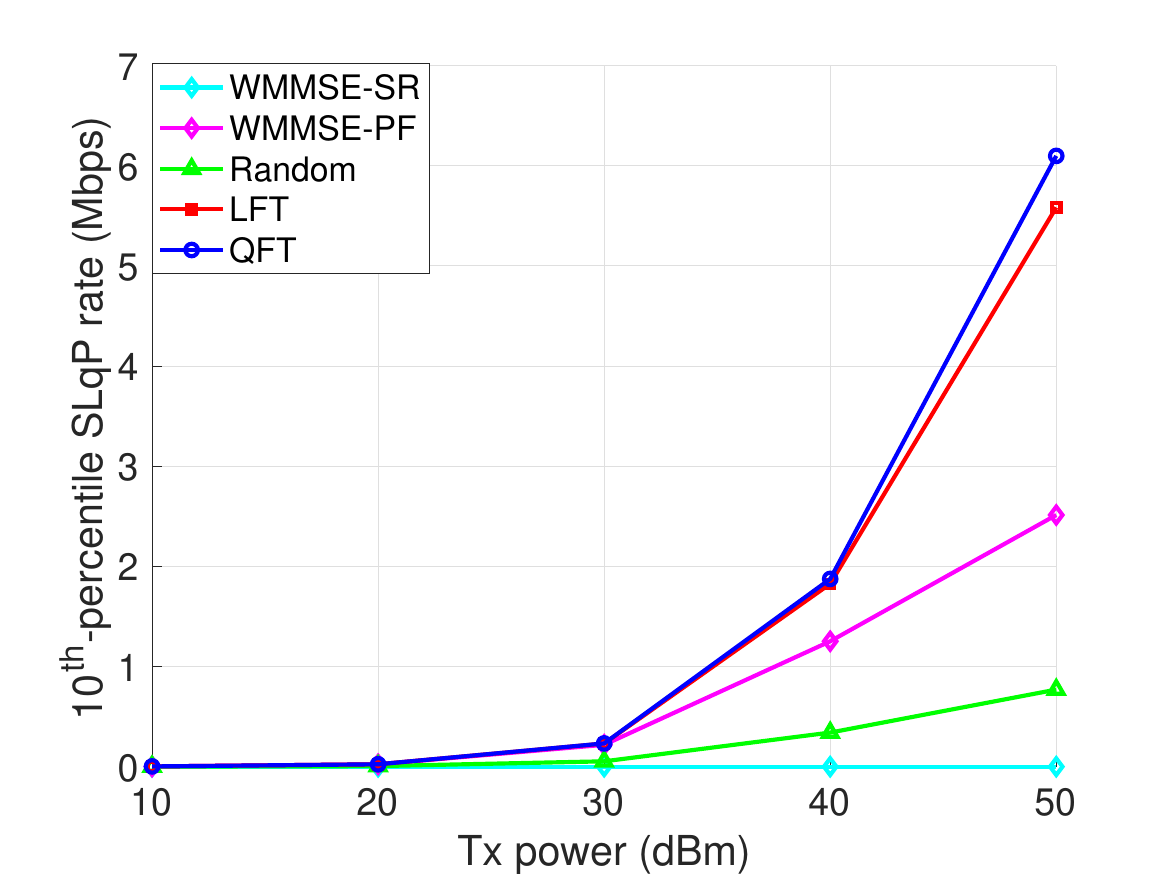}
		\caption{$10^\mathrm{th}$-percentile SLqP rate as a function of $P_\mathrm{max}$ for $K=70$, $K_q=7$.}
		\centering
		\label{power}
	\end{center}
\end{figure} 
\begin{figure}[h!]
	\begin{center} 
		\includegraphics[trim={0cm 0cm 0cm 0cm},clip, width=0.40\textwidth]{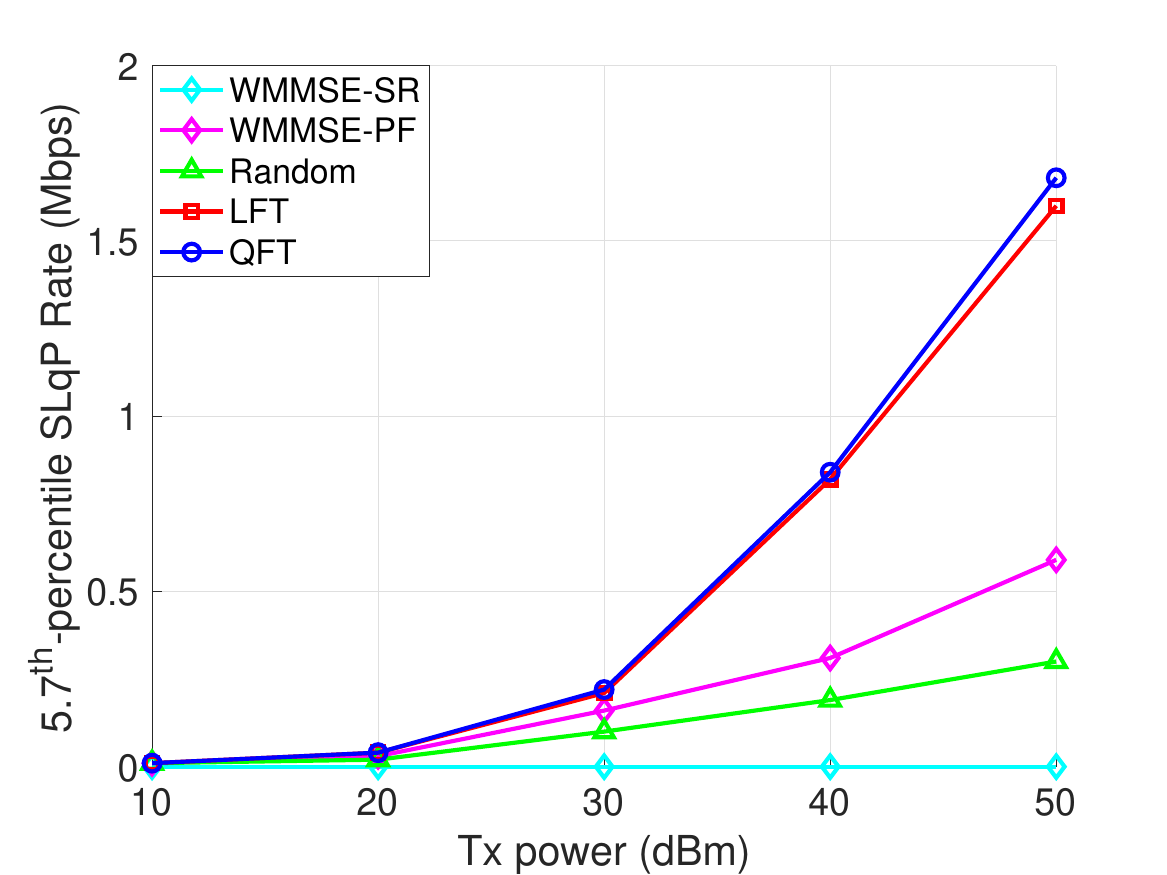}
		\caption{$5.7^\mathrm{th}$-percentile SLqP rate as a function of $P_\mathrm{max}$ for $K=70$, $K_q=4$.}
		\centering
		\label{SLqP_}
	\end{center}
\end{figure}

Finally, it is important to note that while both the QFT and LFT algorithms are minorization-maximization methods they are, in fact, \textit{completely distinct} from a mathematical perspective. This means that even when initialized from the same starting point, in general they do \textit{not} necessarily converge to the same directional stationary point. For example, different algorithms for weighted sum-rate maximization such as FP and WMMSE achieve different levels of performance in different wireless settings \cite{shen_fractional_2018,shen_fractional_2018-1}.

With this context, the closely matched performance of both algorithms is, in itself, a very interesting result. While the reasons for this are not clear, it seems to indicate that improving upon the performance of the QFT and LFT algorithms may require a radically different approach (perhaps one that is not necessarily based on block-concavity inducing transforms).

\section{Conclusions}\label{conclusions}
In this paper, we formulated explicit optimization problems to target sum-percentile throughput maximization through power control in interference-constrained wireless networks. The general class of SLqP throughput problems subsumes the well-known max-min and max-sum-rate optimization problems, and the choice of percentile allows us to flexibly control the inherent tradeoff between favouring cell-center and cell-edge users. With the exception of the two aforementioned extremes of the percentile problems, these problems are non-convex, non-smooth and NP-hard in general. We proposed two iterative algorithms that transform the original intractable problems into block-concave form, thereby enabling guaranteed convergence to stationary points.

Finally, we note that although we have focused on wireless cellular networks in this work, the concept of optimizing metrics at desired percentiles is applicable regardless of the scale, architecture and nature of a communications network.}

\vspace{-1.0em}

\ifCLASSOPTIONcaptionsoff
\newpage
\fi




\bibliographystyle{IEEEtran}
\bibliography{IEEEabrv.bib,biblio.bib}
\vspace{-1.00em}
\appendices
\section{Proof of Theorem \ref{NP_hard_theorem}}\label{NP_hardness_proof}


We now present a proof for the general NP-hardness of the short-term SLqP rate power allocation problem. We extend the technique utilized to prove the strong NP-hardness of the max-sum-rate problem in \cite{zhi-quan_luo_dynamic_2008} to all general SLqP rate power control problems for $q\geq\frac{200}{K}$.

We begin by noting that the rates for users in a single time slot can be expressed in the equivalent form:
\[
r_{k}\left(\mathbf{p}\right)=\mathrm{log}\left(1+\frac{p_{k}}{\sum_{j\neq k}p_{j}\eta_{j\rightarrow k}+\hat{\sigma}_{k}}\right)\]
where, for notational convenience, we have dropped the superscript denoting the current time slot, and 
\[
\eta_{j\rightarrow k}=\frac{\left|h_{j\rightarrow k}\right|^{2}}{\left|h_{k\rightarrow k}\right|^{2}}, \hat{\sigma}_{k}=\frac{\sigma_{k}}{\left|h_{k\rightarrow k}\right|^{2}}
\] 
The coefficient $\eta_{j\rightarrow k}$ can be thought of as the normalized interference channel from transmitter $j$ to receiver $k$; further, we assume $P_\mathrm{max}=1$.

Consider a $K$-node graph $G=(V,E)$, where $V=\left\{ v_{1},v_{2},\ldots,v_{K}\right\}$ denotes the set of vertices and $E$ denotes the set of edges connecting these vertices. Further, suppose that the graph consists of $K-K_q+1$ components defined as follows:
\begin{equation}
C_{i}=\left\{ \begin{array}{cc}
	\left\{ v_{i}\right\}  & i=1,\ldots,K-K_{q}\\
	\left\{ v_{K-K_{q}+1},v_{K-K_{q}+2},\ldots,v_{K}\right\}  & i=K-K_{q}+1
\end{array}\right.
\end{equation}

In other words, the first $K-K_q$ vertices are isolated, while the subgraph induced by the last $K_q$ vertices is connected; Figure \ref{NP_hardness_graph} illustrates an example for $K=10$ and $K_q=5$.

Now, for each $v\in{V}$, we define
\[
\eta_{j\rightarrow k}=\begin{cases}
LK_{q}^{2} & \mathrm{if}\hspace{0.3em}{v_{j}}\hspace{0.3em}\mathrm{is\hspace{0.3em}adjacent\hspace{0.3em}to}\hspace{0.3em}{v_{k}}\\
0 & \mathrm{otherwise}
\end{cases}
\]
where $L>K_q$; further, we set $P_\mathrm{max}=1$ and $\hat{\sigma}_k=L$. The vertices $v_,\ldots,v_{K-K_q}$ can be thought of as cell-center users as they are exempt from interference effects and are therefore capable of achieving the highest rates; on the other hand, users represented by vertices $v_{K-K_q+1},\ldots,v_{K}$ can be thought of as a cell-edge users as they both create interference to, and receive interference from, adjacent vertices within component $\mathcal{C}_{K-K_{q}+1}$. We claim that the SLqP rate problem has an optimal value $f_\mathrm{opt}=\left|I_{\mathcal{C}_{K-K_{q}+1}}\right|\mathrm{log}\left(1+\frac{1}{L}\right)$ if and only if the subgraph induced by the component $\mathcal{C}_{K-K_{q}+1}$ has a maximum independent set of size $\left|I_{\mathcal{C}_{K-K_{q}+1}}\right|$.
\begin{figure}[t!]
	\begin{center} 
		\includegraphics[trim={0.9cm 0cm 0cm 0cm},clip,width=0.4\textwidth]{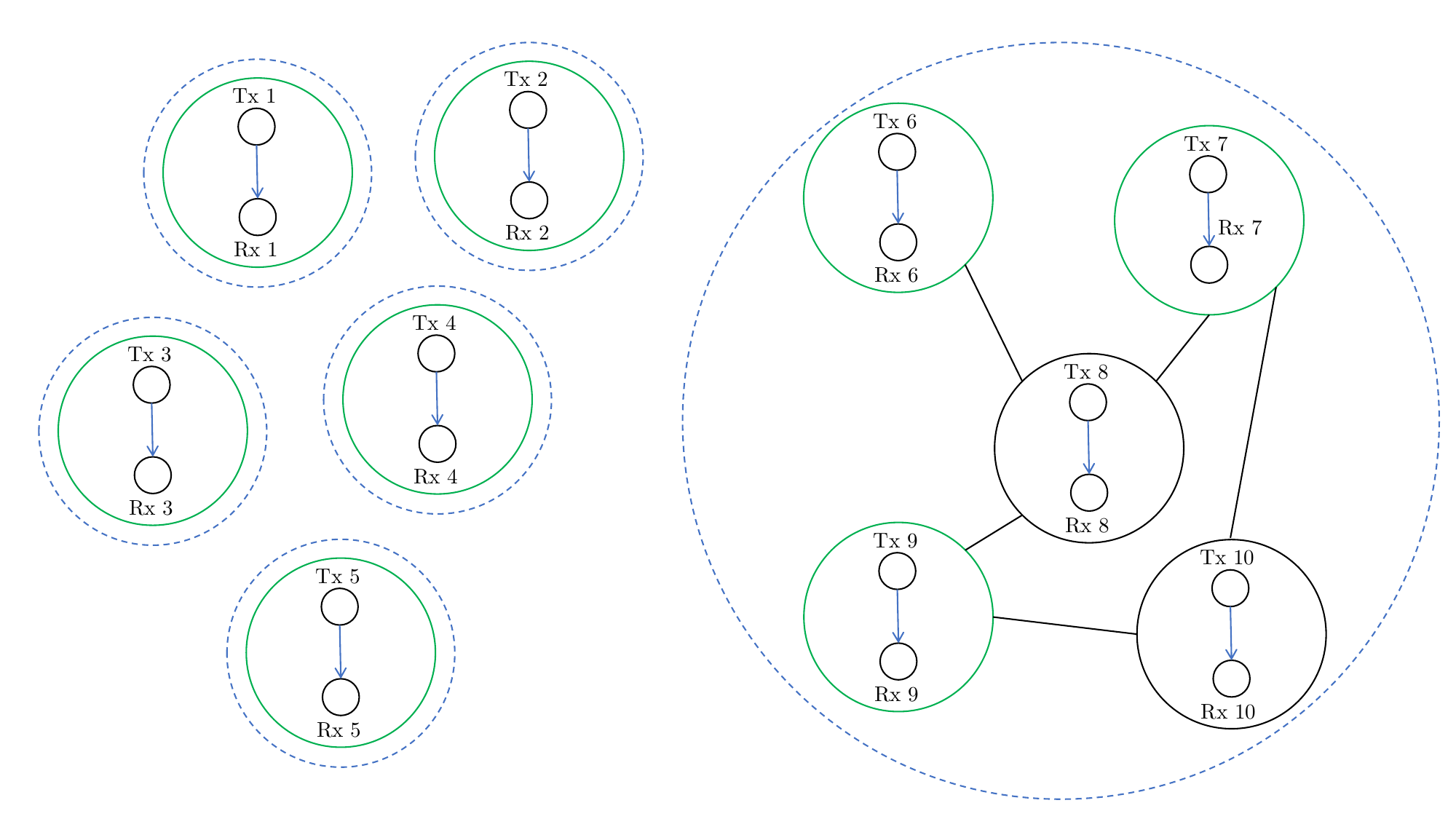}
		\caption{Network interference dynamics for a constructed instance of the SLqP problem with $K=10$ and $K_q=5$. Nodes (i.e., users) that are connected mutually interfere. The users indicated in green transmit at the maximum available power to achieve the highest SLqP ($q=50$) rate.}
		\centering
		\label{NP_hardness_graph}
	\end{center}
\end{figure} 

For achievability, we note that if the subgraph induced by the component $\mathcal{C}_{K-K_{q}+1}$ has a maximum independent set of size $\left|I_{\mathcal{C}_{K-K_{q}+1}}\right|$, then we can set
\[
p_{k}=\begin{cases}
1 & \mathrm{if}\hspace{0.3em}{k=1,\ldots,K-K_{q}}\mathrm{\hspace{0.3em}or\hspace{0.3em}}{v_{k}\in I_{\mathcal{C}_{K-K_{q}+1}}}\hspace{0.3em}\\
0 & \mathrm{otherwise}
\end{cases}
\]
yielding an objective value of $\left|I_{\mathcal{C}_{K-K_{q}+1}}\right|\mathrm{log}\left(1+\frac{1}{L}\right)$.
\begin{figure*}[!h]
	\footnotesize
	\begin{subequations}
		\begin{align}
		\underset{\mathbf{p}^{}}{\mathrm{max}}\;f_{K_{q}}\left(\mathrm{log}\left(1+\frac{A_{1}\left(\mathbf{p}^{}\right)}{B_{1}\left(\mathbf{p}^{}\right)}\right),\ldots,\mathrm{log}\left(1+\frac{A_{K}\left(\mathbf{p}^{}\right)}{B_{K}\left(\mathbf{p}^{}\right)}\right)\right)\hspace{35.3em}\label{quadratic_equiv_proof_line1}\\
		=\underset{\mathbf{p}^{}}{\mathrm{max}}\;f_{K_{q}}\left(\mathrm{log}\left(\underset{x_{1}^{}}{\mathrm{max}}\;\left\{ 1+2x_{1}^{}\sqrt{A_{1}\left(\mathbf{p}^{}\right)}-\left(x_{1}^{}\right)^{2}B_{1}\left(\mathbf{p}^{}\right)\right\} \right),\ldots,\mathrm{log}\left(\underset{x_{K}^{}}{\mathrm{max}}\;\left\{ 1+2x_{K}^{}\sqrt{A_{K}\left(\mathbf{p}\right)}-\left(x_{K}^{}\right)^{2}B_{K}\left(\mathbf{p}\right)\right\} \right)\right)\hspace{0.5500em}    	\label{quadratic_equiv_proof_line2}\\
		=\underset{\mathbf{p}^{}}{\mathrm{max}}\;f_{K_{q}}\left(\underset{x_{1}^{}}{\mathrm{max}}\;\left\{ \mathrm{log}\left(1+2x_{1}^{}\sqrt{A_{1}\left(\mathbf{p}^{}\right)}-\left(x_{1}^{}\right)^{2}B_{1}\left(\mathbf{p}^{}\right)\right)\right\} ,\ldots,\underset{x_{K}^{}}{\mathrm{max}}\;\left\{ \mathrm{log}\left(1+2x_{K}^{}\sqrt{A_{K}\left(\mathbf{p}\right)}-\left(x_{K}^{}\right)^{2}B_{K}\left(\mathbf{p}\right)\right)\right\} \right)\hspace{1.4em}\label{quadratic_equiv_proof_line3}\\
		=\underset{\mathbf{x}^{},\mathbf{p}^{}}{\mathrm{max}}\;f_{K_{q}}\left(\mathrm{log}\left(1+2x_{1}^{}\sqrt{A_{1}\left(\mathbf{p}^{}\right)}-\left(x_{1}^{}\right)^{2}B_{1}\left(\mathbf{p}^{}\right)\right),\ldots,\mathrm{log}\left(1+2x_{K}^{}\sqrt{A_{K}\left(\mathbf{p}\right)}-\left(x_{K}^{}\right)^{2}B_{K}\left(\mathbf{p}\right)\right)\right)\hspace{9.35em}\label{quadratic_equiv_proof_line5}
		\end{align}
	\end{subequations}
	\hrulefill
\end{figure*}

For the converse, suppose that we are given an optimal value $f_\mathrm{opt}$. Observe that since the vertices $v_1,v_2,\ldots,v_{K-K_q}$ are isolated they neither experience nor create interference; therefore, the SLqP objective is non-decreasing in $p_1,p_2,\ldots,p_{K-K_q}$ and these users achieve the $K-K_q$ highest rates in the network. Thus, we can assume that $p_k=1$ for $k=1,\ldots,K-K_{q}$ without loss of generality.

It follows that the SLqP rate on the given graph is therefore equivalent to the sum-rate achieved on the subgraph induced by $\mathcal{C}_{K-K_{q}+1}$. The rest of the derivation follows from the proof of Theorem 1 in \cite{zhi-quan_luo_dynamic_2008}. Denoting the sum-rate achieved on component $\mathcal{C}_{K-K_{q}+1}$ as $f_{K_{q}}^{\mathcal{C}_{K-K_{q}+1}}$, we can compute its second derivative with respect to $p_k$ as:
\begin{multline}\label{Hessian_graph_component}
	\frac{\partial {^2}f_{K_{q}}^{\mathcal{C}_{K-K_{q}+1}}}{\partial p_{k}^{2}}	=-\frac{1}{\left(p_{k}+L+LK_q^{2}\sum_{\left(v_{k},v_{j}\right)\in \mathcal{C}_{K-K_{q}+1}}p_{j}\right)^{2}}\\
	+	\sum_{\left(v_{k},v_{j}\right)\in \mathcal{C}_{K-K_{q}+1}}\frac{1}{\left(K_q^{-2}+\sum_{\left(v_{k},v_{j}\right)\in \mathcal{C}_{K-K_{q}+1}}p_{j}\right)^{2}}
\end{multline}
Since each component of the graph is connected (by definition), the summation in the second term of (\ref{Hessian_graph_component}) is non-zero. Further, since $L>K_q$, the diagonal terms of the Hessian of the sum-rate on the subgraph induced by $\mathcal{C}_{K-K_{q}+1}$ (as given in (\ref{Hessian_graph_component})) are non-negative. Thus, the objective is convex with respect to each individual power variable (although not jointly convex). The maximum of a convex function over a polyhedron is attained at the extreme points of the set. It follows that the optimal power vector $\mathbf{p^*}$ must be a binary vector. Now define
\[
P_{\mathcal{C}_{K-K_{q}+1}}^{*}\coloneqq\left\{ v_{k}\in \mathcal{C}_{K-K_{q}+1}\left|p_{k}=1\right.\right\} 
\]  
Let $\left|I_{\mathcal{C}_{K-K_{q}+1}}\right|$ be the maximum independent set in $\mathcal{C}_{K-K_{q}+1}$. Using the fact that the first $K-K_q$ users attain the maximum rates in the network, we can write the SLqP rate objective equivalently as the sum-rate achieved over $\mathcal{C}_{K-K_{q}+1}$. 

Define $\mathcal{E}_{K-K_{q}+1}$ as the set of connected vertices in $\mathcal{C}_{K-K_{q}+1}$, i.e.,
\begin{equation}
\mathcal{E}_{K-K_{q}+1}\coloneqq\left\{ v_{j}\in\mathcal{C}_{K-K_{q}+1}\left|\left(v_{j},v_{k}\right)\in E\right|\right\} 
\end{equation}

It follows that the sum-rate (and hence optimal SLqP rate objective) over $\mathcal{C}_{K-K_{q}+1}$ can be bounded according to the procedure given in \cite{zhi-quan_luo_dynamic_2008} as:

\begin{subequations}\small\small\small
	\begin{align}
f^\mathrm{opt}=\sum_{v_{k}\in\mathcal{C}_{K-K_{q}+1}}\hspace{-1.30em}\mathrm{log}\left(1+\frac{1}{L+LK_{q}^{2}\left|\left\{ \left(v_{j},v_{k}\right)\in\mathcal{E}_{K-K_{q}+1}\right\} \right|}\right)\label{bound_line1}\hspace{5.30em}\\
<\left|I_{\mathcal{C}_{K-K_{q}+1}}\right|\mathrm{log}\left(1+\frac{1}{L}\right)\hspace{18.500em}
	\end{align}
\end{subequations}

We observe that the number of nodes in the subgraph induced by $\mathcal{C}_{K-K_{q}+1}$ is exactly $K_q$, which grows linearly with the number of nodes $K$ for a fixed value of $q\geq\frac{200}{K}$. Thus, it follows that the SLqP rate is strongly NP-hard since the maximum independent set problem is strongly NP-hard.

Note that the given proof does not hold for the max-min-rate problem; for $K_q=1$ all vertices are isolated and the maximum independent set is $V$ itself. This result is to be expected, since as mentioned earlier, the max-min-rate power allocation problem can be solved to optimality in polynomial time \cite{zhi-quan_luo_dynamic_2008}.
$\blacksquare$
\vspace{-0.700em}
\section{Proof of Theorem \ref{quadratic_SPR_equivalence}}\label{h_quad_proof}

Here we present the proof for the equivalence of the interference-limited SLqP rate maximization problem in (\ref{SPR_problem_shortterm}) with the problem in (\ref{SPR_problem_quad}). The equivalence can be seen by considering the chain of reasoning in (\ref{quadratic_equiv_proof_line1})--(\ref{quadratic_equiv_proof_line5}). Specifically, (\ref{quadratic_equiv_proof_line2}) follows from the result in Theorem (\ref{quadratic_transform}); (\ref{quadratic_equiv_proof_line3}) follows since $\mathrm{log}\left(\cdot\right)$ is a non-decreasing function, so we can move the $\mathrm{max}\left(\cdot\right)$ operator outside it; (\ref{quadratic_equiv_proof_line5}) follows since $f_{K_{q}}\left(\cdot\right)$ is non-decreasing in each component, so we can move the $\mathrm{max}\left(\cdot\right)$ operators corresponding to each component outside it.
\hfill$\blacksquare$
\vspace{-1.00em}
\section{Proof of Theorem \ref{QFT_convergence}}\label{MM_proof}
We begin by reviewing the MM framework as described in \cite{sun_majorization-minimization_2017}. Consider the optimization problem 
\begin{subequations}\label{MM_orig_problem}
	\begin{align}
	\underset{\mathbf{p^{\mathit{}}}}{\mathrm{maximize}}\quad f\left(\mathbf{p}\right)\hspace{5.43em}\\
	\mathrm{subject\,to}\quad\mathbf{p}\in\mathcal{P}\hspace{4.9em}
	\end{align}
\end{subequations}
where $\ensuremath{\ensuremath{f\left(\mathbf{p}\right)}:\mathbb{R}^{n}\mapsto\mathbb{R}}$ is continuous (but not necessarily either smooth or convex), and $\mathcal{P}\subseteq\mathbb{R}^{n}$ is non-empty, closed and convex.

The MM procedure as applied to this problem consists of two alternating steps. At the $i^\mathrm{th}$ iteration, a feasible point $\mathbf{p}\left[i\right]\in\mathcal{P}$ (initialized with some choice of $\mathbf{p}\left[1\right]\in\mathcal{{P}}$) is used to construct a minorizing function $g\left(\mathbf{p}\left|\mathbf{p}\left[i\right]\right.\right)$ such that 
\[
g\left(\mathbf{p}\left|\mathbf{p}\left[i\right]\right.\right)+c\left[i\right]\leq\ensuremath{f\left(\mathbf{p}\right)}\;\forall\mathbf{p}\in\mathcal{P}
\]
where 
\begin{equation}\label{minorizer_gap}
c\left[i\right]=\ensuremath{f\left(\mathbf{p}\left[i\right]\right)}-g\left(\mathbf{p}\left[i\right]\left|\mathbf{p}\left[i\right]\right.\right)
\end{equation}

Next, the minorizing function is maximized to obtain the next feasible iteration point $\mathbf{p}\left[i+1\right]$ as follows:
\[
\mathbf{p}\left[i+1\right]\in\mathrm{arg}\,\underset{\mathbf{p}\in\mathcal{P}}{\mathrm{max}}\quad g\left(\mathbf{p}\left|\mathbf{p}\left[i\right]\right.\right)
\]

Furthermore, define the set of directional stationary points as 
\[
\mathcal{P}^{*}=\left\{ \mathbf{p}\left|f'\left(\mathbf{p}\left[i\right],\mathbf{d}\right)\leq0\quad\forall\left(\mathbf{p}+\mathbf{d}\right)\in\mathcal{P}\right.\right\} 
\]
where $f'\left(\mathbf{p},\mathbf{d}\right)$ denotes the directional derivative of $f$ in the direction $\mathbf{p}$ and is given by
\[
f'\left(\mathbf{p}\left[i\right],\mathbf{d}\right)=\mathrm{lim}\,\underset{t\downarrow0}{\mathrm{inf}}\;\frac{f\left(\mathbf{p}\left[i\right]+t\mathbf{d}\right)-f\left(\mathbf{p}\left[i\right]\right)}{t}
\]
When $f$ is non-smooth, the MM procedure is guaranteed to converge to a point within $\mathcal{P}^{*}$ if the following conditions are satisfied \cite{sun_majorization-minimization_2017}:
\begin{enumerate}\label{convergence_conditions}
	\item The superlevel set $\mathcal{S}_{\geq f\left(\mathbf{p}_{0}\right)}=\left\{ \mathbf{p}\in\mathcal{P}\left|f\left(\mathbf{p}\right)\geq f\left(\mathbf{p}_{0}\right)\right.\right\} $ is compact given that $f\left(\mathbf{p}_{0}\right)>-\infty$.
	\item The directional derivatives of $f$ and $g$ are identical for all directions $\mathbf{d}$ over $\mathcal{{P}}$, i.e.,
	\[
	f'\left(\mathbf{p}\left[i\right],\mathbf{d}\right)=g'\left(\mathbf{p}\left[i\right];\mathbf{d}\left|\mathbf{p}\left[i\right]\right.\right)\quad
	\]
	\item{$g\left(\mathbf{p}\left|\mathbf{p}\left[i\right]\right.\right)$ is continuous in both $\mathbf{p}$ and $\mathbf{p}\left[i\right]$.}
\end{enumerate}
We now demonstrate that the proposed QFT and LFT algorithms are MM algorithms and satisfy these convergence conditions.

Consider the auxiliary objective for the QFT algorithm in (\ref{SPR_obj_quad}); we henceforth denote it by $f_{K_{q}}^{\mathrm{QFT}}\left(\mathbf{p},\mathbf{x}\right)$. We note that the optimal auxiliary variables $\mathbf{x}$ obtained from (\ref{optimal_quadratic_variables}) are a function of $\mathbf{p}$; hence, we denote them subsequently as $\mathbf{x}_\mathrm{opt}\left(\mathbf{p}\right)$. 

Then it follows that for all $\mathbf{p}\left[i\right]\in{\mathcal{P}}$, we must have
\begin{equation}\label{minorizing_inequalities}
f_{K_{q}}^{\mathrm{QFT}}\left(\mathbf{p},\mathbf{x}_\mathrm{opt}\left(\mathbf{p}\left[i\right]\right)\right)\leq f_{K_{q}}^{\mathrm{QFT}}\left(\mathbf{p},\mathbf{x}_\mathrm{opt}\left(\mathbf{p}\right)\right)=f_{Kq}\left(\mathbf{p}\right)
\end{equation}

The inequality follows since the auxiliary objective $f_{K_{q}}^{\mathrm{QFT}}\left(\mathbf{p},\mathbf{x}\right)$ is maximized when the auxiliary variables are updated according to (\ref{optimal_quadratic_variables}) for the given fixed value of $\mathbf{p}$ rather than for $\mathbf{p}\left[i\right]$. The equality holds since the update of the auxiliary variable defined in (\ref{optimal_quadratic_variables}) preserves equality with the original objective function. Thus, it follows that the auxiliary objective $f_{K_{q}}^{\mathrm{QFT}}\left(\mathbf{p},\mathbf{x}\right)$ minorizes $f_{K_{q}}\left(\mathbf{p}\right)$.

Having established that the QFT algorithm is indeed an MM algorithm, we now proceed to verify the conditions in (\ref{convergence_conditions}) to guarantee convergence to a directional stationary point.

\begin{enumerate}
	\item For Problem (\ref{SPR_problem_shortterm}), the set $\mathcal{P}=\left[0,P_{\mathrm{max}}\right]^{K}$ is non-empty, closed, and convex. Next, observe that the SLqP objective in (\ref{SPR_obj_shortterm}) is trivially upper bounded by the interference-free SLqP rate which, in turn, is upper bounded by the interference-free sum-rate from Property \ref{ordering}. Thus, denoting the optimal value of Problem (\ref{SPR_problem_shortterm}) as $f_{K_{q},\mathrm{opt}}$, we must have
	\[
	f_{K_{q},\mathrm{opt}}\leq K{\underset{k=1,\ldots,K}{\mathrm{max}}\mathrm{log}}\left(1+\frac{P_{\mathrm{max}}\left|h_{k\rightarrow k}\right|^{2}}{\sigma^{2}}\right)
	\]
	Taken together with the compactness of $\mathcal{P}$, the continuity and boundedness of the objective in (\ref{SPR_obj_shortterm}) imply that the superlevel sets of Problem (\ref{SPR_problem_shortterm}) are compact.
	\item We observe that (\ref{minorizing_inequalities}) implies that $c\left[i\right]=0$ in (\ref{minorizer_gap}); hence, the auxiliary objective is a so-called \textit{tangent minorant} \cite{jacobson2007expanded}. In particular, this implies that the directional derivative of $f_{K_{q}}^{\mathrm{QFT}}\left(\mathbf{p},\mathbf{x}\right)$ equals that of $f_{K_{q}}\left(\mathbf{p}\right)$ at $\mathbf{p}=\mathbf{p}\left[i\right]$ \cite{jacobson2007expanded}.
	\item We note that $\hat{r}_{k}$ in (\ref{auxiliary_rate_QFT}) is a continuous function of both $\mathbf{p}$ and $\mathbf{x}$ for all $k$. Since the SLqP utility function is continuous and the composition of two continuous functions is also continuous, it follows that the auxiliary QFT objective is also continuous. 
\end{enumerate}

Taken together, these conditions imply the convergence of the QFT algorithm to a directional stationary point of Problem (\ref{SPR_problem_shortterm}). An identical series of arguments can be utilized for the LFT algorithm. $\blacksquare$
\end{document}